\documentclass[a4paper,USenglish]{llncs}

\usepackage{amsmath}
\usepackage{amssymb}
\usepackage{mathtools}
\usepackage{xspace}
\usepackage[utf8]{inputenc}

\usepackage{tikz}
\usetikzlibrary{automata,calc,positioning}
\tikzset{initial text={}}
\tikzset{every initial by arrow/.style={-stealth}}

\usepackage{url}

\newcommand{\quot}[1]{``#1''}

\newcommand{\coloneq}{\mathop{:=}\,}

\renewcommand{\epsilon}{\varepsilon}

\newcommand{\nats}{\mathbb{N}}
\newcommand{\card}[1]{|#1|}

\newcommand{\set}[1]{\{#1\}}
\newcommand{\restr}[2]{{\left.\kern-\nulldelimiterspace #1 \vphantom{\big|} \right|_{#2} }}
\newcommand{\range}{\mathit{range}}
\newcommand{\domain}{\mathit{domain}}
\newcommand{\bin}[1]{\langle#1\rangle}

\newcommand{\aut}{\mathfrak{A}}

\newcommand{\vpdalphabet}{\widetilde{\Sigma}}
\newcommand{\calls}{\Sigma_c}
\newcommand{\returns}{\Sigma_r}
\newcommand{\locals}{\Sigma_l}

\newcommand{\vpa}{\text{VPA}\xspace}
\newcommand{\vpas}{\text{VPAs}\xspace}
\newcommand{\nvpa}{\text{VPA}\xspace}
\newcommand{\tnvpa}{\text{TVPA}\xspace}
\newcommand{\tnvpas}{\text{TVPAs}\xspace}

\newcommand{\oneaja}{\text{1-AJA}\xspace}
\newcommand{\oneajas}{\text{1-AJAs}\xspace}
\newcommand{\bnvpa}{\text{BVPA}\xspace}
\newcommand{\bnvpas}{\text{BVPAs}\xspace}

\newcommand{\rej}{\mathit{rej}}

\newcommand{\push}{\downarrow\!}
\newcommand{\pop}{\uparrow\!}
\newcommand{\local}{\rightarrow}

\newcommand{\graph}{G}

\newcommand{\game}{\mathcal{G}}
\newcommand{\col}{\Omega}
\newcommand{\pdgame}{\mathcal{H}}
\newcommand{\vpsys}{{\mathcal S}}

\newcommand{\halfthinspace}{{\kern .08333em}}

\newcommand{\ttrue}{\texttt{tt}}
\newcommand{\ffalse}{\texttt{ff}}

\newcommand{\ddiamond}[1]{\langle #1 \rangle}
\newcommand{\bbox}[1]{[\halfthinspace #1 \halfthinspace]}

\newcommand{\Raut}{\mathcal{R}}

\newcommand{\ltl}{\text{LTL}\xspace}

\newcommand{\ldl}{\text{LDL}\xspace}

\newcommand{\vldl}{\text{VLDL}\xspace}
\newcommand{\vltl}{\text{VLTL}\xspace}
\newcommand{\dpldl}{\text{DPLDL}\xspace}
\newcommand{\caret}{\text{CaRet}\xspace}

\newcommand{\omegavpl}{\ensuremath{\omega}\text{-VPL}\xspace}

\newcommand{\pspace}{\textsc{{PSpace}}\xspace}
\newcommand{\exptime}{\textsc{{ExpTime}}\xspace}
\newcommand{\twoexp}{\textsc{{2ExpTime}}\xspace}
\newcommand{\threeexp}{\textsc{{3ExpTime}}\xspace}

\newcommand{\steps}{\mathit{steps}}
\newcommand{\sh}{\mathit{sh}}
\newcommand{\stepsindex}{\mathit{st}}

\newcommand{\ini}{0}
\newcommand{\fin}{1}

\newcommand{\comms}{\mathit{Comms}}

\newcommand{\ltlnext}{{\mathbf X}}
\newcommand{\ltluntil}{{\mathbf U}}

\newcommand{\ltleventually}{{\mathbf F}}

\newcommand{\tm}{{\mathcal T}}

\newcommand{\bplus}{\mathcal{B}^{+}}

\newcommand{\vps}{\text{VPS}\xspace}
\newcommand{\vpss}{\text{VPS's}\xspace}
\newcommand{\traces}{\mathit{traces}}

\begin{document}

\title{Visibly Linear Dynamic Logic
\thanks{Supported by the projects ``TriCS'' (ZI 1516/1-1) and ``AVACS'' (SFB/TR 14) of the German Research Foundation (DFG).}
}
\author{Alexander Weinert and Martin Zimmermann}

\institute{Reactive Systems Group, Saarland University, 66123 Saarbrücken, Germany\\
	\email{\{weinert, zimmermann\}@react.uni-saarland.de}}

\maketitle

\begin{abstract}
We introduce Visibly Linear Dynamic Logic (\vldl), which extends Linear Temporal Logic~(LTL) by temporal operators that are guarded by visibly pushdown languages over finite words.
In \vldl one can, e.g., express that a function resets a variable to its original value after its execution, even in the presence of an unbounded number of intermediate recursive calls.
We prove that \vldl describes exactly the $\omega$-visibly pushdown languages.
Thus it is strictly more expressive than \ltl and able to express recursive properties of programs with unbounded call stacks.

The main technical contribution of this work is a translation of \vldl into $\omega$-visibly pushdown automata of exponential size via one-way alternating jumping automata.
This translation yields exponential-time algorithms for satisfiability, validity, and model checking.
We also show that visibly pushdown games with \vldl winning conditions are solvable in triply-exponential time.
We prove all these problems to be complete for their respective complexity classes.
\end{abstract}


\section{Introduction}
\label{sec:introduction}
Linear Temporal Logic (\ltl) \cite{pnueli77} is widely used for the specification of non-terminating systems.
Its popularity is owed to its simple syntax and intuitive semantics, as well as to the so-called exponential compilation property, i.e., for each \ltl formula there exists an equivalent Büchi automaton of exponential size.
Due to the latter property, there exist algorithms for model checking in polynomial space and for solving infinite games in doubly-exponential time.

While \ltl suffices to express properties of circuits and non-recursive programs with bounded memory, its application to real-life programs is hindered by its inability to express recursive properties.
In fact, \ltl is too weak to even express all $\omega$-regular properties.
There are several approaches to address the latter shortcoming by augmenting \ltl, e.g., with regular expressions \cite{leucker07,vardi11}, finite automata on infinite words~\cite{vardi94}, and right-linear grammars~\cite{wolper83}.
We concentrate on the approach of Linear Dynamic Logic (\ldl) \cite{vardi11}, which guards the globally- and eventually-operators of \ltl with regular expressions.
While the \ltl-formula $\ltleventually\psi$ simply means \quot{either now, or at some point in the future, $\psi$ holds}, the corresponding \ldl operator $\ddiamond{r}\psi$ means \quot{There exists an infix matching the regular expression~$r$ starting at the current position, and $\psi$ holds true after that infix}.

The logic \ldl captures the $\omega$-regular languages.
In spite of its greater expressive power, \ldl still enjoys the exponential compilation property, hence there exist algorithms for model checking and solving infinite games in polynomial space and doubly-exponential time, respectively.

While the expressive power of \ldl is sufficient for many specifications, it is still not able to reason about recursive properties of systems.
In order to address this shortcoming, we replace the regular expressions guarding the temporal operators with visibly pushdown languages (VPLs)~\cite{alur04} specified by visibly pushdown automata (\vpas)~\cite{alur04}.

A \vpa is a pushdown automaton that operates over a fixed partition of the input alphabet into calls, returns and local actions.
In contrast to classical pushdown automata, \vpas may only push symbols onto the stack when reading calls and may only pop symbols off the stack when reading returns.
Moreover, they may not even inspect the topmost symbol of the stack when not reading returns.
Thus, the height of the stack after reading a word is known a priori for all \vpas using the same partition of the input alphabet.
Due to this, \vpas are closed under union and intersection, as well as complementation.
The class of languages accepted by \vpas is known as visibly pushdown languages.

The class of such languages over infinite words, i.e., $\omega$-visibly pushdown languages, are known to allow for the specification of many important properties in program verification such as \quot{there are infinitely many positions at which at most two functions are active}, which expresses repeated returns to a main-loop, or \quot{every time the program enters a module $m$ while $p$ holds true, $p$ holds true upon exiting $m$}~\cite{alur04}.
The extension of \vpas to their variant operating on infinite words is, however, not well-suited to the specification of such properties in practice, as Boolean operations on such automata do not preserve the logical structure of the original automata.
By guarding the temporal operators introduced in~\ldl with \vpas, \vldl allows for the modular specification of recursive properties while capturing $\omega$-\vpas.

\subsection{Our contributions}

We begin with an introduction of \vldl and give examples of its use. 

Firstly, we provide translations from \vldl to \vpas over infinite words, so-called $\omega$-\vpas, and vice versa.
For the direction from logic to automata we translate \vldl formulas into one-way alternating jumping automata (1-AJA), which are known to be translatable into $\omega$-\vpas of exponential size due to Bozzelli \cite{bozzelli07}.
For the direction from automata to logic we use a translation of $\omega$-\vpas into deterministic parity stair automata by Löding et al.~\cite{loeding04}, which we then translate into \vldl formulas.

Secondly, we prove the satisfiability problem and the validity problem for \vldl to be \exptime-complete.
Membership in \exptime follows from the previously mentioned constructions, while we show \exptime-hardness of both problems by a reduction from the word problem for polynomially space-bounded alternating Turing machines adapting a similar reduction by Bouajjani et al.~\cite{bouajjani97}.

As a third result, we show that model checking visibly pushdown systems against \vldl specifications is \exptime-complete as well.
Membership in \exptime follows from \exptime-membership of the model checking problem for \oneajas against visibly pushdown systems.
\exptime-hardness follows from \exptime-hardness of the validity problem for \vldl.

Moreover, solving visibly pushdown games with \vldl winning conditions is proven to be \threeexp-complete.
Membership in \threeexp follows from the exponential translation of \vldl formulas into $\omega$-\vpas and the membership of solving pushdown games against $\omega$-\vpa winning conditions in \twoexp due to Löding et al. \cite{loeding04}.
\threeexp-hardness is due to a reduction from solving pushdown games against \ltl specifications, again due to Löding~et~al.~\cite{loeding04}.

Finally, we show that replacing the visibly pushdown automata used as guards in \vldl by deterministic pushdown automata yields a logic with an undecidable satisfiability problem.

Our results show that \vldl allows for the concise specification of important properties in a logic with intuitive semantics.
In the case of satisfiability and model checking, the complexity jumps from \pspace-completeness for \ldl to \exptime-completeness.
For solving infinite games, the complexity gains an exponent moving from \twoexp-completeness to \threeexp-completeness.

We choose \vpas for the specification of guards in order to simplify arguing about the expressive power of \vldl.
In order to simplify the modeling of $\omega$-VPLs, other formalisms that capture VPLs over finite words may be used.
We discuss one such formalism in the conclusion.

\subsection{Related Work}

The need for specification languages able to express recursive properties has been identified before and there exist other approaches to using visibly pushdown languages over infinite words for specifications, most notably \vltl~\cite{bozzelli14b} and \caret \cite{alur04b}.
While \vltl captures the class of $\omega$-visibly pushdown languages, \caret captures only a strict subset of it.
For both logics there exist exponential translations into $\omega$-\vpas.
In this work, we provide exponential translations from \vldl to $\omega$-\vpas and vice versa.
Hence, \caret is strictly less powerful than \vldl, but every \caret formula can be translated into an equivalent \vldl formula, albeit with a doubly-exponential blowup.
Similarly, every \vltl formula can be translated into an equivalent \vldl formula and vice versa, with doubly-exponential blowup in both directions.

In contrast to \vltl, which augments regular expressions with substitution operators (replacing occurrences of local actions by well-matched words), \vldl instead extends the well-known concepts introduced for \ltl and \ldl with visibly pushdown automata.
Hence, specifications written in \vldl are  modular and have an intuitive semantics, in particular for practitioners already used to~\ltl.

Other logical characterizations of visibly pushdown languages include characterizations by a fixed-point logic \cite{bozzelli07} and by monadic second order logic augmented with a binary matching predicate (MSO$_\mu$) \cite{alur04}.
Even though these logics also capture the class of visibly pushdown languages, they feature neither an intuitive syntax nor intuitive semantics and thus are less applicable than \vldl in a practical setting.

\section{Preliminaries}
\label{sec:preliminaries}

In this section we introduce the basic notions used in the remainder of this work. 
A pushdown alphabet $\vpdalphabet = (\calls, \returns, \locals)$ is a finite set $\Sigma$ that is partitioned into calls $\calls$, returns $\returns$, and local actions $\locals$.
We write $w = w_0\cdots w_n$ and $\alpha = \alpha_0\alpha_1\alpha_2\cdots$ for finite and infinite words, respectively.
The stack height $\sh(w)$ reached after reading $w$ is defined inductively as $\sh(\epsilon) = 0$, $\sh(wc) = \sh(w) + 1$ for $c \in \calls$, $\sh(wr) = \max \set{0, \sh(w) - 1}$ for $r \in \returns$, and $\sh(wl) = \sh(w)$ for $l \in \locals$.
We say that a call $c \in \calls$ at some position $k$ of a word $w$ is matched if there exists a $k' > k$ with $w_{k'} \in \returns$ and $\sh(w_0\cdots w_k) - 1 = \sh(w_0\cdots w_{k'})$.
The return at the earliest such position $k'$ is called the matching return of $c$.
We define $\steps(\alpha) \coloneq \set{ k \in \nats \mid \forall k' \geq k.\  \sh(\alpha_0\cdots\alpha_{k'}) \geq \sh(\alpha_0\cdots\alpha_k)}$ as the positions reaching a lower bound on the stack height along the remaining suffix.
Note that we have $0 \in \steps(\alpha)$ and that $\steps(\alpha)$ is infinite for infinite words~$\alpha$.

\paragraph{Visibly Pushdown Systems}
A visibly pushdown system (\vps) $\vpsys = (Q, \vpdalphabet, \Gamma, \Delta)$ consists of a finite set~$Q$ of states, a pushdown alphabet $\vpdalphabet$, a stack alphabet $\Gamma$, which contains a stack-bottom marker $\bot$, and a transition relation \[ \Delta \subseteq (Q \times \calls \times Q \times (\Gamma \setminus \set{\bot})) \cup (Q \times \returns \times \Gamma \times Q) \cup (Q \times \locals \times Q)\enspace.\]
Note that we write the stack content from top to bottom, i.e., the tip of the stack is on the left.
A configuration $(q, \gamma)$ of $\vpsys$ is a pair of a state $q \in Q$ and a stack content $\gamma \in \Gamma_c = (\Gamma \setminus \set{\bot})^* \cdot \bot$.
The VPS $\vpsys$ induces the configuration graph $\graph_\vpsys = (Q \times \Gamma_c, E)$ with $E \subseteq ((Q \times \Gamma_c) \times \Sigma \times (Q \times \Gamma_c))$ and $((q, \gamma), a, (q', \gamma')) \in E$  if, and only if, either
\begin{enumerate}
\item $a \in \calls$, $(q, a, q', A) \in \Delta$, and $A\gamma = \gamma'$,
\item $a \in \returns$, $(q, a, \bot, q') \in \Delta$, and $\gamma = \gamma' = \bot$,
\item $a \in \returns$, $(q, a, A, q') \in \Delta$, $A \neq \bot$, and $\gamma = A\gamma'$, or
\item $a \in \locals$, $(q, a, q') \in \Delta$, and $\gamma = \gamma'$.
\end{enumerate}
For an edge $e = ((q, \gamma), a, (q', \gamma'))$, we call~$a$ the label of~$e$.
We define a run $\pi = (q_0, \gamma_0)\cdots(q_n, \gamma_n)$ of $\vpsys$ on $w = w_0\cdots w_{n-1}$ as a sequence of configurations where $((q_i, \gamma_i), w_i, (q_{i+1}, \gamma_{i+1})) \in E$ in $\graph_\vpsys$ for all $i \in [0;n-1]$.
The VPS $\vpsys$ is deterministic if for each vertex $(q, \gamma)$ in $\graph_\vpsys$ and each $a \in \Sigma$ there exists at most one outgoing $a$-labeled edge from $(q, \gamma)$.
In figures, we write $\push A$, $\pop A$ and~$\local$ to denote pushing~$A$ onto the stack, popping~$A$ off the stack, and local actions, respectively.

\paragraph{(Büchi) Visibly Pushdown Automata}
A visibly pushdown automaton (\vpa) \cite{alur04} is a six-tuple $\aut = (Q, \vpdalphabet, \Gamma, \Delta, I, F)$, where $\vpsys = (Q, \vpdalphabet, \Gamma, \Delta)$ is a VPS and $I, F \subseteq Q$ are sets of initial and final states.
A run $(q_0, \gamma_0)(q_1, \gamma_1)(q_2, \gamma_2)\cdots$ of $\aut$ is a run of~$\vpsys$, which we call initial if $(q_0, \gamma_0) = (q_I, \bot)$ for some $q_I \in I$.
A finite run $\pi = (q_0, \gamma_0)\cdots(q_n, \gamma_n)$  is accepting if $q_n \in F$.
A Büchi \vpa (\bnvpa) is syntactically identical to a \vpa, but we only consider runs over infinite words.
An infinite run is Büchi-accepting if it visits states in $F$ infinitely often.
A (B)VPA $\aut$ accepts a word~$w$ (an infinite word $\alpha$) if there exists an initial (Büchi-)accepting run of $\aut$ on $w$ ($\alpha$).
We denote the family of languages accepted by (B)VPA by ($\omega$-)VPL.

\section{Visibly Linear Dynamic Logic}
\label{sec:vldl}

We fix a finite set $P$ of atomic propositions and a partition $\vpdalphabet = (\calls, \returns, \locals)$ of $2^P$ throughout this work.
The syntax of \vldl is defined by the grammar
\[
\varphi \coloneq p \mid \neg \varphi \mid \varphi \wedge \varphi \mid \varphi \vee \varphi
  \mid \ddiamond{\aut} \varphi 
  \mid \bbox{\aut} \varphi \enspace,
 \]
where $p \in P$ and where~$\aut$ ranges over testing visibly pushdown automata (\tnvpa) over~$\vpdalphabet$.
We define a \tnvpa $\aut = (Q, \vpdalphabet, \Gamma, \Delta, I, F, t)$ as consisting of a \nvpa $(Q, \vpdalphabet, \Gamma, \Delta, I, F)$ and a partial function~$t$ mapping states to \vldl formulas over $\vpdalphabet$.\footnote{Obviously, there are some restrictions on the nesting of tests into automata. More formally, we require the subformula relation to be acyclic as usual.}
Intuitively, such an automaton accepts an infix $\alpha_i\cdots\alpha_j$ of an infinite word $\alpha_0\alpha_1\alpha_2\cdots$ if the embedded \nvpa has an initial accepting run $(q_i, \gamma_i)\cdots(q_{j+1}, \gamma_{j+1})$ on $\alpha_i\cdots\alpha_j$ such that, if $q_{i+k}$ is marked with $\varphi$ by $t$, then $\alpha_{i+k}\alpha_{i+k+1}\alpha_{i+k+2}\cdots$ satisfies $\varphi$.

We define the size of $\varphi$ as the sum of the number of subformulas (including those contained as tests in automata and their subformulas) and  of the numbers of states of the automata contained in $\varphi$.
As shorthands, we use $\ttrue \coloneq p \lor \neg p$ and $\ffalse \coloneq p \land \neg p$ for some atomic proposition $p$.
Even though the testing function $t$ is defined as a partial function, we generally assume it is total by setting $t\colon q \mapsto \ttrue$ if $q \notin \domain(t)$.

Let $\alpha = \alpha_0\alpha_1\alpha_2\cdots$ be an infinite word over $2^P$ and let $k \in \nats$ be a position in $\alpha$.
We define the semantics of \vldl inductively via
\begin{itemize}

\item $(\alpha, k) \models p$ if, and only if, $p \in \alpha_k$,

\item $(\alpha, k) \models \neg \varphi$ if, and only if, $(\alpha, k) \not\models \varphi$,

\item $(\alpha, k) \models \varphi_0 \land \varphi_1$ if, and only if, $(\alpha, k) \models \varphi_0$ and $(\alpha, k) \models \varphi_1$, and dually for $\varphi_0 \lor \varphi_1$,

\item $(\alpha, k) \models \ddiamond{\aut}\varphi$ if, and only if, there exists $l \geq k$ s.t.\ $(k, l) \in \Raut_\aut(\alpha)$ and $(\alpha, l) \models \varphi$,

\item $(\alpha, k) \models \bbox{\aut}\varphi$ if, and only if, for all $l \geq k$, $(k, l) \in \Raut_\aut(\alpha)$ implies $(\alpha, l) \models \varphi$,

\end{itemize}
where $\Raut_\aut(\alpha)$ contains all pairs of positions $(k, l)$ such that $\aut$ accepts $\alpha_k\cdots\alpha_{l-1}$.
Formally, we define
\begin{multline*}
	\Raut_\aut(\alpha) \coloneq \{ (k, l) \in \nats \times \nats \mid \exists \text{ initial accepting run } (q_k, \sigma_k)\cdots (q_l, \sigma_l) \text{ of } \aut \\ \text{ on } \alpha_k\cdots\alpha_{l-1}
	\textbf{and } \forall m \in \set{k,\dots,l}.\ (\alpha, m) \models t(q_m) \}.
\end{multline*}
We write $\alpha \models \varphi$ as a shorthand for $(\alpha, 0) \models \varphi$ and say that $\alpha$ is a model of~$\varphi$ in this case.
The language of $\varphi$ is defined as $L(\varphi) \coloneq \set{ \alpha \in (2^P)^\omega \mid \alpha \models \varphi}$.
As usual, disjunction and conjunction are dual, as well as the $\ddiamond{\aut}$-operator and the $\bbox{\aut}$-operator, which can be dualized using De Morgan's law and the logical identity $\bbox{\aut}\varphi \equiv \neg \ddiamond{\aut}\neg \varphi$, respectively.
Note that the latter identity only dualizes the temporal operator, but does not require complementation of the automaton guarding the operator.
We additionally allow the use of derived boolean operators such as~$\rightarrow$ and~$\leftrightarrow$, as they can easily be reduced to the basic operators~$\land$,~$\lor$ and~$\neg$.

The logic \vldl combines the expressive power of visibly pushdown automata with the intuitive temporal operators of \ldl.
Thus, it allows for concise and intuitive specifications of many important properties in program verification~\cite{alur04}.
In particular, \vldl allows for the specification of recursive properties, which makes it more expressive than both \ldl \cite{vardi11} and \ltl \cite{pnueli77}.
In fact, we can embed \ldl in \vldl in linear time.



\begin{lemma}
\label{lem:ldl-to-vldl}
For any \ldl formula $\psi$ over $P$ we can effectively construct a \vldl formula~$\varphi$ over $\vpdalphabet \coloneq (\emptyset, \emptyset, 2^P)$ in linear time such that $L(\psi) = L(\varphi)$.	
\end{lemma}

\begin{proof}
	We define $\varphi$ by structural induction over $\psi$.
	The only interesting case is $\psi = \ddiamond{r}\psi'$, since all other cases follow from closure properties and duality.
	We obtain the \vldl formula $\varphi'$ over $\vpdalphabet$ equivalent to $\psi'$ by induction and construct the finite automaton $\aut_r$ from $r$ using the construction of Faymonville and Zimmermann~\cite{FaymonvilleZimmermann17}.
	The automaton $\aut_r$ contains tests, but is not equipped with a stack.
	Since $\vpdalphabet = (\emptyset, \emptyset, 2^P)$, we can interpret $\aut_r$ as a \tnvpa without changing the language it recognizes.
	We call the \tnvpa $\aut'_r$ and define $\varphi = \ddiamond{\aut'_r}\varphi'$.
\end{proof}

Since \ltl can be in turn embedded in \ldl in linear time, Lemma \ref{lem:ldl-to-vldl} directly implies the embeddability of \ltl in \vldl in linear time.
Note that this proof motivates the use of \tnvpas instead of \vpas without tests as guards in order to obtain a concise formalism.
We later show that removing tests from these automata does not change the expressiveness of \vldl.
It is, however, open whether it is possible to translate even \ltl formulas into \vldl formulas without tests in polynomial time.

%

\section{Examples of \vldl Specifications}
As we will show in Section \ref{sec:expressiveness}, \vldl captures the visibly pushdown languages and thus, it is strictly more expressive than traditional Büchi automata.
In fact, \vldl allows for concise formulations of a number of important properties of recursive programs in program verification.
We give some examples of such properties and their formalization in this section.

\begin{example}
\label{ex:vldl}
	Assume that we have a program that may call some module $m$ and that has the observable atomic propositions $P \coloneq \set{c, r, p, q}$, where $c$ and $r$ denote calls to and returns from $m$, and $p$ and $q$ are arbitrary propositions.

	We now construct a \vldl formula that describes the condition \quot{If $p$ holds true immediately after entering $m$, it shall hold immediately after the corresponding return from $m$ as well}~\cite{alur04b}.
	For the sake of readability, we assume that the program never emits both $c$ and $r$ in the same step.
	Moreover, we assume that the program emits at least one atomic proposition in each step.
	Since we want to count the calls and returns occurring in the program using the stack, we pick the pushdown alphabet $\vpdalphabet = (\calls, \returns, \locals)$ such that $P' \subseteq P$ is in~$\calls$ if $c \in P'$, $P' \in \returns$ if $r \in P'$, but $c \notin P'$, and $P' \in \locals$ otherwise.
	
	The formula $\varphi \coloneq \bbox{\aut_c} (p \rightarrow \ddiamond{\aut_r}p)$ then captures the condition above, where Figure~\ref{fig:vldl-example} shows $\aut_c$ and $\aut_r$.
	The automaton $\aut_c$ accepts all finite words ending with a call to $m$, whereas the automaton $\aut_r$ accepts all words ending with a single unmatched return.
	
	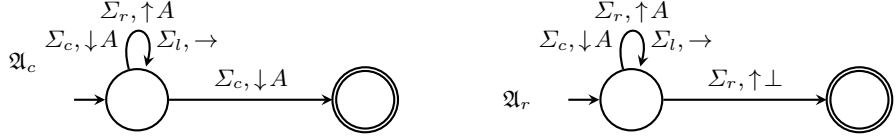
\begin{figure}[h]
	\centering
		\begin{tikzpicture}[thick]
			\node at (0, 0.5) {$\aut_c$};
			\node[state,initial] (call-init) at (1.5, 0) {};
			\node[state,accepting] (call-final) at (4.5, 0) {};
			
			\path[->,>=stealth] (call-init) edge [loop above] node [anchor=east,pos=.25] {$\calls, \push A$ }
				node [anchor = south,pos=.5] {$\returns, \pop A$}
				node [anchor=west,pos=.75] {$\locals, \local$} 
				(call-init);
			\path[->,>=stealth] (call-init) edge node [anchor=south] {$\Sigma_c, \push A$} (call-final);
			
			\node at (6.5,0) {$\aut_r$};
			\node[state,initial] (return-init) at (8, 0) {};
			\node[state,accepting] (return-final) at (11, 0) {};
			
			\path[->,>=stealth] (return-init) edge [loop above]
				node [anchor=east,pos=.25] {$\calls, \push A$ }
				node [anchor = south,pos=.5] {$\returns, \pop A$}
				node [anchor=west,pos=.75] {$\locals, \local$} (return-init);
			\path[->,>=stealth] (return-init) edge node [anchor=south] {$\Sigma_r, \pop \bot$} (return-final);
		\end{tikzpicture}
		\caption{The automata $\aut_c$ and $\aut_r$ for Example \ref{ex:vldl}. We depict the final states with a double circle.}
		\label{fig:vldl-example}
	\end{figure}
	
	Figure~\ref{fig:example-bnvpa} shows a \bnvpa $\aut$ describing the same specification as $\varphi$.
	For the sake of readability, we use $\Sigma_x^p = \set{P' \in \Sigma_x \mid p \in P'}$ and $\Sigma_x^{\neg p} = \set{P' \in \Sigma_x \mid p \notin P'}$ for $x \in \set{c,r,l}$.
	In contrast to $\varphi$, which uses only a single stack symbol, namely $A$, the \bnvpa $\aut$ has to rely on the two stack symbols~$P$ and~$\bar{P}$ to track whether or not $p$ held true after entering the module $m$.
	Moreover, there is no direct correlation between the logical structure of the specification and the structure of the \bnvpa, which exemplifies the difficulty of maintaining specifications given as \bnvpas.
	
	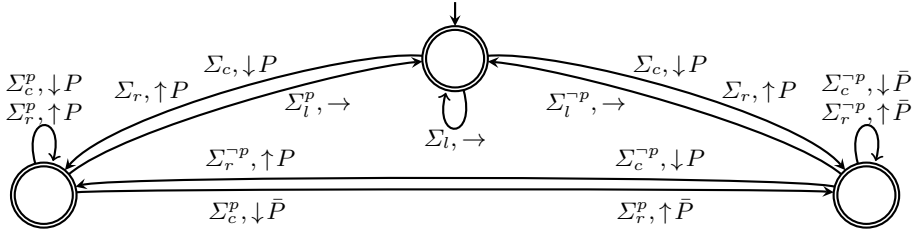
\begin{figure}[h]
		\centering
			\begin{tikzpicture}[thick,-stealth,node distance=1.2cm and 4.8cm]
				\node[state, accepting] (initial) at (0,0) {};
				\node[state, accepting,below left=of initial] (p) {};
				\node[state, accepting,below right=of initial] (q) {};
				
				\path (0,.75) edge (initial);
				
				\path (initial)
					edge [out=175,in=50,looseness=.5]
					node [anchor=south east,inner sep=0,pos=.4] {
						$\calls, \push P$
					} node [anchor=south east,inner sep=0,pos=.6] {
						$\returns, \pop P$
					} (p);
				\path (initial)
					edge [out=5,in=130,looseness=.5]
					node [anchor=south west,inner sep=0,pos=.4] {
						$\calls, \push P$
					} node [anchor=south west,inner sep=0,pos=.6] {
						$\returns, \pop P$
					} (q);
					
				\path (initial)
					edge [loop below]
					node [anchor=north,inner sep=0] {
						$\locals, \local$
					} (initial);
				
				\path (p)
					edge [out=40,in=-175,looseness=.5]
					node [anchor=north west,inner sep=0,pos=.6] {
						$\Sigma_l^p, \local$
					} (initial);
				\path (p)
					edge [out=5,in=175,looseness=.15]
					node [anchor=north,align=center,pos=.3] {
						$\calls^p, \push \bar{P}$
					} 
					node [anchor=north,align=center,pos=.7] {
						$\returns^p, \pop \bar{P}$
					} (q);
				\path (p)
					edge [loop above]
					node [anchor=south,inner sep=0,align=left] {
						$\calls^p, \push P$\\
						$\returns^p, \pop P$
					} (p);
				
				\path (q)
					edge [out=140,in=-5,looseness=.5]
					node [anchor=north east,inner sep=0,pos=.6] {
						$\locals^{\neg p}, \local$
					} (initial);
				\path (q)
					edge [out=165,in=15,looseness=.15]
					node [anchor=south,pos=.3] {
						$\Sigma_c^{\neg p}, \push P$
					} node [anchor=south,pos=.7] {
						$\Sigma_r^{\neg p}, \pop P$
					} (p);
				\path (q)
					edge [loop above]
					node [anchor=south,inner sep=0,align=left] {
						$\calls^{\neg p}, \push \bar P$\\
						$\returns^{\neg p}, \pop \bar P$
					} (q);
			\end{tikzpicture}
		\caption{A \bnvpa $\aut$ specifying the same language as $\varphi$ from Example~\ref{ex:vldl}.}
		\label{fig:example-bnvpa}
	\end{figure}
	
	Finally, one could also specify the property above using the formalism of \vltl~\cite{bozzelli14b}, which eschews automata in favor of augmented regular expression.
	One such formula would be $\psi \coloneq (\alpha;\ttrue) |\alpha \rangle \ffalse$, where the visibly rational expression $\alpha$ is defined as
	\[ \alpha \coloneq \left[ (p\cup q)^*c \left[ (q \square) \cup (p \square p)\right]
	r(p\cup q)^* \right]^{\circlearrowleft_\square} \curvearrowleft_\square (p\cup q)^* \]
	that uses the additional local action $\square$.
	Again, the conditional nature of the specification is lost in the translation to \vltl.
	Moreover, the temporal nature is not well visible in the formal specification due to use of the non-standard operator $\psi|\alpha\rangle\psi$.
	
	In contrast to these two alternative formal specifications, \vldl offers a readable and intuitive formalism that combines the well-known standard acceptors for visibly pushdown languages with guarded versions of the widely used temporal operators of \ltl and the readability of classical logical operators.
\end{example}

Note that the stack is simply used as a counter in Example \ref{ex:vldl}.
This technique suffices for the specification of other properties as well, such as tracking the path through the directory structure instead of the call stack.

\begin{example}
\label{ex:directories}
We consider a simplified system model, in which a user can move through directories and obtain and relinquish superuser rights.
To this end, we consider the set of atomic propositions $P = \set{\texttt{cd}_\downarrow, \texttt{cd}_\uparrow, \texttt{sudo}, \texttt{logout}}$, where $\texttt{cd}_\downarrow$ denotes moving into a subdirectory of the current working directory, $\texttt{cd}_\uparrow$ denotes moving to the parent directory, $\texttt{sudo}$ denotes the acquisition of elevated privileges, and $\texttt{logout}$ denotes relinquishing them.
For readability, we only define the pushdown alphabet for singleton subsets of $P$ and pick the alphabet $\vpdalphabet \coloneq (\set{\texttt{cd}_\downarrow}, \set{\texttt{cd}_\uparrow}, \set{\texttt{sudo}, \texttt{logout}})$ in order to formalize the property \quot{If the program acquires elevated privileges, it has to relinquish them before moving out of its current directory} \cite{chen02}.

We use the stack as a counter using the stack alphabet $\Gamma \coloneq \set{\bot, A}$.
Then the formula $\varphi \coloneq \bbox{\aut_\mathit{priv}}\bbox{\aut_\mathit{par}}\ffalse$, specifies the property above, where $\aut_\mathit{priv}$ accepts all prefixes of runs of the program that end with the acquisition of elevated privileges, and $\aut_\mathit{par}$ tracks the depth of the current working directory.
Figure \ref{fig:directories-example} depicts the automata $\aut_\mathit{priv}$ and $\aut_\mathit{par}$.

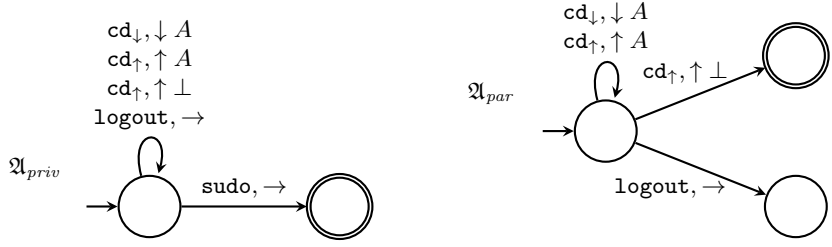
\begin{figure}[h]
	\begin{center}
		\begin{tikzpicture}[thick]
			
			\node at (0, 0.5) {$\aut_\mathit{priv}$};
			\node[state,initial] (priv-init) at (1.5, 0) {};
			\node[state,accepting] (priv-final) at (4, 0) {};
			
			\path[->,>=stealth] (priv-init) edge [loop above] node [anchor=south,align=center] {$\texttt{cd}_\downarrow, \downarrow A$\\$\texttt{cd}_\uparrow, \uparrow A$\\$\texttt{cd}_\uparrow, \uparrow \bot$\\$\texttt{logout}, \rightarrow$} (priv-init);
			\path[-stealth] (priv-init) edge node [anchor=south] {$\texttt{sudo}, \rightarrow$} (priv-final);
			
			\node at (6, 1.5) {$\aut_\mathit{par}$};
			
			\node[state,initial] (par-init) at (7.5, 1) {};
			\node[state,accepting] (par-acc) at (10, 2) {};
			\node[state] (par-rej) at (10, 0) {};
			
			\path[->,>=stealth] (par-init) edge [loop above] node [anchor=south,align=center] {$\texttt{cd}_\downarrow, \downarrow A$\\$\texttt{cd}_\uparrow, \uparrow A$} (par-init);
			
			\path[-stealth] (par-init) edge node [anchor=south,align=center,transform canvas={xshift=-.2cm}] {$\texttt{cd}_\uparrow, \uparrow \bot$} (par-acc);
			\path[-stealth] (par-init) edge node [anchor=north,align=center,transform canvas={xshift=-.4cm}] {$\texttt{logout}, \rightarrow$} (par-rej);
			
		\end{tikzpicture}
		\caption{The automata $\aut_{\mathit{priv}}$ and $\aut_\mathit{par}$ for Example \ref{ex:directories}.}
		\label{fig:directories-example}
	\end{center}
\end{figure}
\end{example}

While the previous example shows how to handle programs that can simply request a single set of elevated rights, in actual systems the situation is more complicated.
In reality, a program may request the rights of any user of the system by logging in as that user.
When logging out, the rights revert to those of the previously logged in user.
In the following example we use the stack to keep track of the currently logged in user and ensure that system calls are not executed with elevated privileges.

\begin{example}
\label{ex:user}
We remove some of the simplifications of the previous example and model the login mechanism of an actual system more precisely.
	To this end, let $P = \set{\texttt{exec}, \texttt{login}_s, \texttt{login}_u, \texttt{logout}}$, where $\texttt{exec}$ denotes the execution of a system call, $\texttt{login}_s$ and $\texttt{login}_u$ denote the login as the superuser and some other user, respectively, and $\texttt{logout}$ denotes logging the current user out and reverting to the previous user.
	The pushdown alphabet defined as $\vpdalphabet \coloneq (\set{\texttt{login}_s, \texttt{login}_u}, \set{\texttt{logout}}, \set{\texttt{exec}})$ allows us to keep track of the stack of logged in users.
	We want to specify the property \quot{Whenever the program has obtained elevated privileges, it does not leave the directory it originally obtained these privileges in before relinquishing them.}
	
	Recall that visibly pushdown automata are not allowed to inspect the top of the stack.
	Thus, in order to correctly trace the currently logged in user, we need to store both the current user and the previously logged in user on the stack.
	The automaton $\aut_\mathit{user}$ performs this bookkeeping using the stack alphabet $\Gamma \coloneq \set{(c, p) \mid c, p \in \set{s, u}}$, where $c$ denotes the currently logged in user, and $p$ denotes the previously logged in user.
	It moves to the state $u$ when a normal user is logged in and to state $s$ when a superuser is logged in.
	
	\begin{figure}[h]
		\begin{center}
		\begin{tikzpicture}[thick]
			\node at (0, 0.5) {$\aut_\mathit{user}$};
			\node[state, initial] (user-init) at (1.5, 0) {$u$};
			\node[state, accepting] (user-final) at (6.5, 0) {$s$};
			
			\path[->,>=stealth] (user-init) edge [loop above] node [anchor=south, align=center] {$\texttt{exec}, \rightarrow$\\$\texttt{login}_u, \push(u, u)$\\$\texttt{logout}, \pop(u, u)$} (user-init);
			\path[->,>=stealth] (user-final) edge [loop above] node [anchor=south, align=center] {$\texttt{exec}, \rightarrow$\\$\texttt{login}_s, \push(s, s)$\\$\texttt{logout}, \pop(s, s)$} (user-final);
			
			\path[-stealth] (user-init) edge [bend angle=15, bend left] node [anchor=south, align=center] {$\texttt{login}_s, \push (s, u)$\\$\texttt{logout}, \pop (u, s)$} (user-final);
			\path[-stealth] (user-final) edge [bend angle=15, bend left] node [anchor=north, align=center] {$\texttt{login}_u, \push (u, s)$\\$\texttt{logout}, \pop (s, u)$} (user-init);
		\end{tikzpicture}	
		\caption{The automaton $\aut_\mathit{user}$, which keeps track of the status of the currently logged in user.}
		\label{fig:user-example}
		\end{center}

	\end{figure}
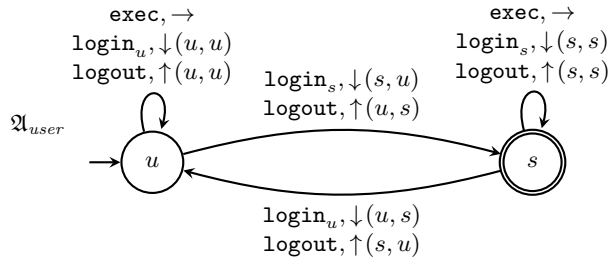
	
	Since the only action available to the program in this example apart from logging users in or out is to execute system calls, we do not need an additional automaton to capture the undesired behavior, but can simply use the atomic proposition $\texttt{exec}$ in the formula.
	Hence, the formula $\varphi \coloneq \bbox{\aut_\mathit{user}}\neg\texttt{exec}$ defines the desired behavior.
\end{example}

Due to the modular nature of \vldl, we can easily reuse existing automata and subformulas.
Consider, e.g., a setting similar to that of Examples~\ref{ex:directories} and~\ref{ex:user} with the added constraint that we want to make sure that superusers neither execute system calls, nor leave the directory they were in when they acquired superuser-privileges.
Using some simple modifications to $\aut_\mathit{user}$ and $\aut_\mathit{par}$ to work over an extended set of atomic propositions, we can specify the conjunction of the previously defined behaviors without having to construct new automata from scratch.

\section{\vldl Captures $\boldsymbol{\omega}$-VPL}
\label{sec:expressiveness}

In this section we show that \vldl characterizes $\omega$-VPL.
Recall that a language is in $\omega$-VPL if, and only if, there exists a \bnvpa recognizing it.
We provide effective constructions transforming \bnvpas into equivalent \vldl formulas and vice versa.

\begin{theorem}
\label{thm:vldl-eq-vpa}
	For any language of infinite words $L \subseteq \Sigma^\omega$ there exists a \bnvpa $\aut$ with $L(\aut) = L$ if, and only if, there exists a \vldl formula $\varphi$ with $L(\varphi) = L$.
	There exist effective translations for both directions.
\end{theorem}

In Section \ref{sec:expressiveness:stair-automata-to-vldl} we show the construction of \vldl formulas from \bnvpas via deterministic parity stair automata.
In Section \ref{sec:expressiveness:vldl-to-aja} we construct one-way alternating jumping automata from \vldl formulas.
These automata are known to be translatable into equivalent \bnvpas.
Both constructions incur an exponential blowup in size.
We show this blowup to be unavoidable in the construction of \bnvpas from \vldl formulas.
It remains open whether the blowup can be avoided in the construction for the other direction.

\subsection{From Stair Automata to VLDL}
\label{sec:expressiveness:stair-automata-to-vldl}

In this section we construct a \vldl formula of exponential size that is equivalent to a given \bnvpa $\aut$.
To this end, we first transform $\aut$ into an equivalent deterministic parity stair automaton (DPSA)  \cite{loeding04} in order to simplify the translation.
A PSA $\aut = (Q, \vpdalphabet, \Gamma, \Delta, I, \col)$ consists of
a VPS $\vpsys = (Q, \vpdalphabet,\Gamma, \Delta)$,
a set of initial states $I$,
and a coloring $\col\colon Q \rightarrow \nats$.
The automaton $\aut$ is deterministic if $\vpsys$ is deterministic and if $\card{I} = 1$.

A run of $\aut$ on a word $\alpha$ is a run of the underlying \vps $\vpsys$ on $\alpha$.
Recall that a step is a position at which the stack height reaches a lower bound for the remainder of the word.
A stair automaton only evaluates the parity condition at the steps of the word.
Formally, a run $\rho_\alpha = (q_0, \sigma_0)(q_1, \sigma_1)(q_2, \sigma_2)\cdots$ of~$\aut$ on the word $\alpha$ induces a sequence of colors $\col(\rho_\alpha) \coloneq \col(q_{k_0})\col(q_{k_1})\col(q_{k_2})\cdots$, where $k_0 < k_1 < k_2 \cdots$ is the ordered enumeration of the steps of $\alpha$.
A DPSA~$\aut$ accepts an infinite word~$\alpha$ if there exists an initial run $\rho$ of $\aut$ on $\alpha$ such that the largest color appearing infinitely often in $\col(\rho)$ is even.
The language $L(\aut)$ of a parity stair automaton $\aut$ is the set of all words $\alpha$ that are accepted by $\aut$.

\begin{lemma}
\label{lem:nvpa-to-stair}
	For every \bnvpa $\aut$ there exists an effectively constructible equivalent DPSA~$\aut_\stepsindex$ with $\card{\aut_\stepsindex} \in \mathcal O(2^{\card{\aut}})$ \cite{loeding04}.
\end{lemma}

Since the stair automaton $\aut_\stepsindex$ equivalent to a \bnvpa $\aut$ is deterministic, the acceptance condition collapses to the requirement that the unique run of $\aut_\stepsindex$ on $\alpha$ must be accepting.
Another important observation is that every time $\aut_\stepsindex$ reaches a step of $\alpha$, the stack may be cleared:
Since the topmost element of the stack will never be popped after reaching a step, and since \vpas cannot inspect the top of the stack, neither this symbol, nor the ones below it have any influence on the remainder of the run.

Thus, the formula equivalent to $\aut_\stepsindex$ has to specify the following constraints: There must exist some state~$q$ of even color such that the stair automaton visits~$q$ at a step, afterwards the automaton may never visit a higher color again at a step, and each visit to~$q$ at a step must be followed by another visit to~$q$ at a step.
All of these conditions can be specified by \vldl formulas in a straightforward way, since $\aut_\stepsindex$ is deterministic and since there is only a finite number of colors in $\aut_\stepsindex$.

\begin{lemma}
\label{lem:stair-to-vldl}
For each DPSA $\aut$ there exists an effectively constructible equivalent \vldl formula $\varphi_{\aut}$ with $\card{\varphi_{\aut}} \in \mathcal O(\card{\aut}^2)$.
\end{lemma}

\begin{proof}
We first construct a formula $\varphi_\stepsindex$ such that, for each word~$\alpha$, we have $(\alpha, k) \models \varphi_\stepsindex$ if, and only if, $k \in \steps(\alpha)$:
Let $\aut_\stepsindex$ be a \vpa that accepts upon reading an unmatched return, constructed similarly to $\aut_r$ from Example~\ref{ex:vldl}.
Then we can define $\varphi_\stepsindex \coloneq \bbox{\aut_\stepsindex}\ffalse$, i.e., we demand that the stack height never drops below the current level by disallowing $\aut_\stepsindex$ to accept any prefix.

Let~$Q$ and~$\col$ be the state set and the coloring function of~$\aut$, respectively.
In the remainder of this proof, we write $_{I'}\aut_{F'}$ to denote the \tnvpa that we obtain from combining the \vps of $\aut$ with the sets  $I'$ and $F'$ of initial and final states.
Additionally, we require that $_{I'}\aut_{F'}$ does not accept the empty word.
This is trivially true if the intersection of $I'$ and $F'$ is empty, and easily achieved by adding a new initial state if it is not.
Furthermore, we define $Q_\mathit{even} \coloneq \set{ q \in Q \mid \col(q) \text{ is even}}$ and $Q_{>q} \coloneq \set{ q' \in Q \mid \col(q') > \col(q)}$.

Recall that $\aut$ accepts a word $\alpha$ if the largest color seen infinitely often at a step during the run of $\aut$ on $\alpha$ is even.
This is equivalent to the existence of a state $q$ as characterized above.
These conditions are formalized as
	\[\varphi_1(q) \coloneq \ddiamond{_{I}\aut_{\set{q}}} ( \varphi_{\stepsindex} \land \bbox{_{\set{q}}\aut_{Q_{>q}}}\neg\varphi_\stepsindex )\]
and
	\[\varphi_2(q) \coloneq \bbox{_{I}\aut_{\set{q}}}(\varphi_\stepsindex \rightarrow \ddiamond{_{\set{q}}\aut_{\set{q}}}\varphi_\stepsindex)\enspace ,\] respectively.
We obtain $\varphi_\aut \coloneq \bigvee_{q \in Q_\mathit{even}} ( \varphi_1(q) \land \varphi_2(q))$.

The construction of $\varphi_2(q)$ relies heavily on the determinism of the DPSA~$\aut$.
If $\aut$ were not deterministic, the universal quantification over all runs ending in $q$ at a step would also capture eventually rejecting partial runs.
Since there only exists a single run of $\aut$ on the input word, however, $\varphi_\aut$ has the intended meaning.
Furthermore, both $\varphi_1(q)$ and $\varphi_2(q)$ use the observation that we are able to clear the stack every time that we reach a step.
Thus, although the stack contents are not carried over between the different automata, the concatenation of the automata does not change the resulting run.
Hence, we have $\alpha \in L(\aut)$  if, and only if, $(\alpha, 0) \models \varphi_\aut$ and thus $L(\aut) = L(\varphi_\aut)$.
\end{proof}

Combining Lemmas \ref{lem:nvpa-to-stair} and \ref{lem:stair-to-vldl} yields that \vldl is at least as expressive as \bnvpa.
The construction inherits an exponential blowup from the construction of DPSAs from \bnvpas and proves one direction of Theorem~\ref{thm:vldl-eq-vpa}.

In the next section we show that each \vldl formula can be transformed into an equivalent \nvpa of exponential size.
Thus, the construction from the proof of Lemma~\ref{lem:stair-to-vldl} yields a normal form for \vldl formulas. In particular, formulas in this normal form only use temporal operators up to nesting depth three.

\begin{proposition}
\label{prop:vldl-normal-form}
Let $\varphi$ be a \vldl formula.
There exists an equivalent formula 
$\varphi'= \bigvee\nolimits_{i=1}^n ( \ddiamond{\aut_i^1} ( \varphi_{\stepsindex} \land \bbox{\aut_i^2}\neg\varphi_\stepsindex ) \land \bbox{\aut_i^1}(\varphi_\stepsindex \rightarrow \ddiamond{\aut_i^3}\varphi_\stepsindex))$,
for some $n$ that is doubly-exponential in $\card{\varphi}$, where all $\aut_i^j$ share the same underlying \vps, $\varphi_\stepsindex$ is fixed over all $\varphi$, and neither the $\aut_i^j$ nor $\varphi_\stepsindex$ contain tests.
\end{proposition}

Proposition~\ref{prop:vldl-normal-form} shows that tests are syntactic sugar.
However, removing them incurs a doubly-exponential blowup.
It remains open whether this blowup can be avoided. 

\subsection{From \vldl to \oneaja}
\label{sec:expressiveness:vldl-to-aja}

We now construct, for a given \vldl formula, an equivalent \bnvpa.
A direct construction would incur a non-elementary blowup due to the unavoidable exponential blowup of complementing \bnvpas.
Moreover, it would be difficult to handle runs of the \vpas over finite words and their embedded tests, which run in parallel.
Thus, we extend a construction by Faymonville and Zimmermann \cite{FaymonvilleZimmermann17}, where a similar challenge was addressed using alternating automata.
Instead of alternating visibly pushdown automata, however, we use one-way alternating jumping automata (1-AJA) , which can be translated into equivalent \bnvpas of exponential size~\cite{bozzelli07}.

A 1-AJA $\aut = (Q, \vpdalphabet, \delta, I, \col)$ consists of
	a finite state set~$Q$,
	a visibly pushdown alphabet~$\vpdalphabet$,
	a transition function~$\delta\colon Q \times \Sigma \rightarrow \bplus( \comms_Q )$,
	where $\comms_Q \coloneq \set{\rightarrow, \rightarrow_a} \times Q \times Q$, with $\bplus( \comms_Q )$ denoting the set of positive Boolean formulas over $\comms_Q$,
	a set $I \subseteq Q$ of initial states,
	and a coloring $\col\colon Q \rightarrow \nats$.
We define $\card{\aut} = \card{Q}$.
Intuitively, when the automaton is in state~$q$ at position $k$ of the word $\alpha = \alpha_0\alpha_1\alpha_2\cdots$, it guesses a set of commands $R \subseteq \comms_Q$ that is a model of $\delta(q, \alpha_k)$.
It then spawns one copy of itself for each command $(d, q, q')\in R$ and executes the command with that copy.
If $d = \rightarrow_a$ and if $\alpha_k$ is a matched call, the copy jumps to the position of the matching return of $\alpha_k$ and transitions to state $q'$.
Otherwise, i.e., if $d = \rightarrow$, the automaton advances to position $k+1$ and transitions to state $q$.
All copies of~$\aut$ proceed in parallel.
A single copy of $\aut$ accepts if the highest color visited infinitely often is even.
A \oneaja accepts $\alpha$ if all of its copies accept.


\begin{lemma}
\label{lem:aja-to-nvpa}
For every \oneaja~$\aut$ there exists an effectively constructible equivalent \bnvpa~$\aut_\mathit{vp}$ with $\card{\aut_\mathit{vp}} \in \mathcal O(2^{\card{\aut}})$ \cite{bozzelli07}.
\end{lemma}

For a given \vldl formula $\varphi$ we now inductively construct a \oneaja  that recognizes the same language as $\varphi$.
The main difficulty lies in the translation of formulas of the form $\ddiamond{\aut}\varphi$, since these require us to translate \tnvpas over finite words into \oneajas over infinite words.
We do so by adapting the idea for the translation from \bnvpas to \oneajas by Bozzelli~\cite{bozzelli07} and by combining it with the bottom-up translation from \ldl into alternating automata by Faymonville and Zimmermann~\cite{FaymonvilleZimmermann17}.

\begin{lemma}
\label{lem:vldl-to-aja}
For any \vldl formula $\varphi$ there exists an effectively constructible equivalent \oneaja $\aut_\varphi$ with $\card{\aut_\varphi} \in  \mathcal{O}(\card{\varphi}^2)$.
\end{lemma}

\begin{proof}
We construct the automaton inductively over the structure of $\varphi$.
	The case $\varphi = p$ is trivial.
	For Boolean operations, we obtain $\aut_\varphi$ by closure of \oneajas under these operations~\cite{bozzelli07}.
	If $\varphi = \bbox{\aut}\varphi'$ we use the identity $\bbox{\aut}\varphi' \equiv \neg\ddiamond{\aut}\neg\varphi'$ and construct $\aut_{\neg\ddiamond{\aut}\neg\varphi'}$ instead.

	We now consider $\varphi = \ddiamond{\aut}\varphi'$, where $\aut$ is some \tnvpa and construct a 1-AJA~$\aut_\varphi$.
	By induction we obtain a \oneaja $\aut'$ equivalent to $\varphi'$.
	$\aut_\varphi$ simulates a run of $\aut$ on a prefix of~$\alpha$ and, upon acceptance, nondeterministically transitions into~$\aut'$.
	
	Consider an initial run of $\aut$ on a prefix $w$.
	Since $w$ is finite, $\steps(w)$ is finite as well.
	Hence, each stack height may only be encountered finitely often at a step.
	At the last visit to a step of a given height, $\aut$ either accepts, or it reads a call action.
	The symbol pushed onto the stack in that case does not influence the remainder of the run.
	We show such a run on the word $clcrrcclrll$ in Figure~\ref{fig:nvpa-run}, where $c$ is a call, $r$ is a return, and $l$ is a local action.
	
	\begin{figure}[htbp]
	\centering
	\begin{tikzpicture}[thick,xscale=.9,yscale=.9]
		
		\foreach \letter [count=\position] in {c,l,c,r,r,c,c,l,r,l,l} {
			\node at ($(\position - .5, 5)$) {\letter};
		}
		\foreach \state in {0,...,11} {
			\node at ($(\state, 4.5)$) {$q_{\state}$};
		}
		\foreach \position in {3,7,8} {
			\node at ($(\position, 4)$) {$B$};
		}
		\foreach \position in {1,2,3,4,6,7,8,9,10,11} {
			\node at ($(\position, 3.6)$) {$A$};
		}
		\foreach \position in {0,...,11} {
			\node at ($(\position, 3.2)$) {$\bot$};
		}

		\node at (-.5,5)   {$\alpha$};
		\node at (-.5,4.5) {$q$};
		\node at (-.5,3.5) {$\gamma$};
		
		\foreach \ypos in {5,4.5,3.5} {
			\node at (12,\ypos) {$\cdots$};
		}
		
		\foreach \ycoord in {4.25,4.75} {
			\draw[thin,black] (-.75,\ycoord) -- (11.5,\ycoord);
		}
		\foreach \xcoord in {-.25} {
			\draw[thin,black] (\xcoord,3) -- (\xcoord,5.25);
		}
	\end{tikzpicture}
	\caption{Run of a \nvpa on the word $clcrrcclrll$.}
	\label{fig:nvpa-run}
	
	\end{figure}
	
	The idea for the simulation of the run of $\aut$ by~$\aut_\varphi$ is to have a main copy of~$\aut_\varphi$ that jumps along the steps of the input word.
	When $\aut_\varphi$ encounters a call $c \in \calls$ it guesses whether or not $\aut$ encounters the current stack height again.
	If it does, then $\aut_\varphi$ guesses $q', q'' \in Q$ and $A \in \Gamma$ such that $(q, c, q', A)$ is a transition of $\aut$, it jumps to the matching return of $c$ with state $q''$ and it spawns a copy that verifies that $\aut$ can go from the configuration $(q', A)$ to the configuration $(q'', \bot)$.
	If~$\aut$ never returns to the current stack height, then~$\aut_\varphi$ only guesses $q' \in Q$ and $A \in \Gamma$ such that $(q, c, q', A)$ is a transition of $\aut$, moves to state $q'$, and stores in its state space that it may not read any returns anymore.
	This is repeated until the main copy guesses that $\aut'$ accepts the prefix read so far.
	
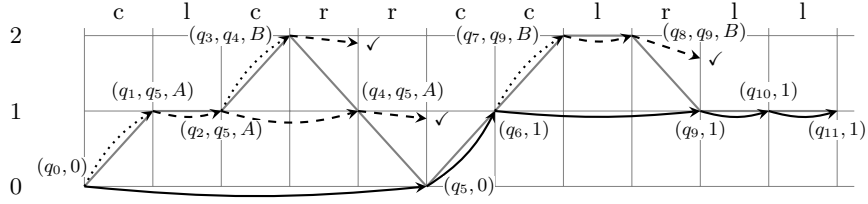
\begin{figure}[htbp]
	\centering
	\begin{tikzpicture}[thick,xscale=.9]
		
		\newcommand{\labelscale}{.8}
		\newcommand{\shcolor}{gray}
		\newcommand{\maincolor}{black}
		\newcommand{\vericolor}{black}
		\newcommand{\veristyle}{dashed}
	
		\newcommand{\jump}[3]{\draw[thick,#3] (#1) edge [bend right=15,-stealth] (#2);}
		\newcommand{\spawn}[2]{\draw[thick,#2,dotted] (#1) edge [bend left = 15,-stealth] ($(#1) + (1,1)$);}
		\newcommand{\accept}[2]{
			\draw[thick,#2] (#1) edge[-stealth] ($(#1) + (1,-.1)$);
			\node[text=\vericolor,anchor=west,scale=\labelscale] at ($(#1) + (1,-.1)$) {$\checkmark$};
		}
		\newcommand{\mylabel}[4]{\node[text=#3,anchor=#4,scale=\labelscale,fill=white,inner sep=0,outer sep=.5em] at (#1) {#2};}
		\newcommand{\nelabel}[3]{\mylabel{#1}{#2}{#3}{south west}}
		\newcommand{\nwlabel}[3]{\mylabel{#1}{#2}{#3}{south east}}
		\newcommand{\selabel}[3]{\mylabel{#1}{#2}{#3}{north west}}
		\newcommand{\swlabel}[3]{\mylabel{#1}{#2}{#3}{north east}}
		\newcommand{\northlabel}[3]{\mylabel{#1}{#2}{#3}{south}}
		\newcommand{\southlabel}[3]{\mylabel{#1}{#2}{#3}{north}}
		\newcommand{\eastlabel}[3]{\mylabel{#1}{#2}{#3}{west}}
		\newcommand{\westlabel}[3]{\mylabel{#1}{#2}{#3}{east}}
	
		\foreach \letter [count=\position] in {c,l,c,r,r,c,c,l,r,l,l} {
			\node[anchor=south] at ($(\position - .5, 2.1)$) {\letter};
		}
		
		\foreach \heightlabel in {0,1,2} {
			\node at (-1,\heightlabel) {\heightlabel};
		}
		
		\draw[step=1.0,gray,very thin] (0,0) grid (11.25,2.1);
		
		\def\lastheight{0}
		\foreach \height [count=\position,remember=\height as \lastheight] in {1,1,2,1,0,1,2,2,1,1,1} {
			\draw[\shcolor] ($(\position - 1, \lastheight)$) -- (\position, \height);
		}
		
		\draw[thick,\maincolor] (0,0) edge [bend right=5,-stealth] (5,0);
		\draw[thick,\maincolor] (5,0) edge [bend right=15,-stealth] (6,1);
		\draw[thick,\maincolor] (6,1) edge [bend right=5,-stealth] (9,1);
		\draw[thick,\maincolor] (9,1) edge [bend right = 15,-stealth] (10,1);
		\draw[thick,\maincolor] (10,1) edge [bend right = 15,-stealth] (11,1);
		
		\def\lastx{1}
		\def\lasty{1}
		\foreach \x/\y [remember=\x as \lastx,remember=\y as \lasty] in {2/1,4/1} {
			\jump{\lastx,\lasty}{\x,\y}{\vericolor,\veristyle}
		}
		
		\def\lastx{7}
		\def\lasty{2}
		\foreach \x/\y [remember=\x as \lastx,remember=\y as \lasty] in {8/2} {
			\jump{\lastx,\lasty}{\x,\y}{\vericolor,\veristyle}
		}
		
		\spawn{0,0}{\maincolor}
		\spawn{2,1}{\vericolor}
		\spawn{6,1}{\maincolor}
		
		\accept{3,2}{\vericolor,\veristyle}
		\accept{4,1}{\vericolor,\veristyle}
		\draw[thick,\vericolor,\veristyle] (8,2) edge[-stealth] (9,1.7);
		\node[text=\vericolor,anchor=west,scale=\labelscale] at (9,1.7) {$\checkmark$};
		
		\nwlabel{.2,0}{$(q_0, \ini)$}{\maincolor}
		\northlabel{1,1}{$(q_1, q_5, A)$}{\vericolor}
		\southlabel{2,1}{$(q_2, q_5, A)$}{\vericolor}
		\westlabel{2.9,2}{$(q_3, q_4, B)$}{\vericolor}
		\nelabel{3.9,1}{$(q_4, q_5, A)$}{\vericolor}
		\eastlabel{5.1,0}{$(q_5, \ini)$}{\maincolor}
		\selabel{5.9,1}{$(q_6, \fin)$}{\maincolor}
		\westlabel{6.8,2}{$(q_7, q_9, B)$}{\vericolor}
		\eastlabel{8.3,2.05}{$(q_8, q_9, B)$}{\vericolor}
		\southlabel{9,1}{$(q_9, \fin)$}{\maincolor}
		\northlabel{10,1}{$(q_{10}, \fin)$}{\maincolor}
		\southlabel{11,1}{$(q_{11}, \fin)$}{\maincolor}
		
	\end{tikzpicture}
		
	\caption{Simulation of the run from Figure~\ref{fig:nvpa-run} by a \oneaja.}
	\label{fig:stack-height}	
	\end{figure}
		
	We show the run of such a \oneaja corresponding to the run of $\aut$ in Figure~\ref{fig:nvpa-run} is shown in Figure \ref{fig:stack-height}.
	The gray line indicates the stack height, while the solid and dashed black paths denote the run of the main automaton and those of the verifying automata, respectively.
	Dotted lines indicate spawning a verifying automaton.
	For readability, the figure does not include copies of the automata that are spawned to verify that the tests of $\aut$ hold true.
	The main copy of the automaton uses states of the form~$(q,0)$ if it has not yet ignored any call actions, and states of the form~$(q,1)$ if it has done so.
	The states $(q, q', A)$ denote verification copies that verify $\aut$'s capability to move from the configuration $(q,A)$ to the configuration $(q',\bot)$.
	The verification automata work similarly to the main automaton, except that they assume that all pushed symbols to be eventually popped and reject if they encounter an unmatched call.
	We now construct the \oneaja~$\aut_\varphi$ equivalent to~$\ddiamond{\aut}\varphi'$ formally.
	
Let $\aut = (Q^\aut, \vpdalphabet, \Gamma^\aut, \Delta^\aut, I^\aut, F^\aut, t^\aut)$, let $\aut' = (Q', \vpdalphabet, \delta', I', \col')$ be the \oneaja equivalent to~$\varphi'$ and, for each $\varphi_i \in \range(t^\aut)$, let $\aut_i = (Q^i, \vpdalphabet, \delta^i, I^i, \col^i)$ be a \oneaja equivalent to~$\varphi_i$.
The automata~$\aut'$ and~$\aut_i$ are obtained by induction.
	
Formally, we use the set of states 
\[ Q \coloneq (Q^\aut \times \{\ini, \fin\}) \cup (Q^\aut \times Q^\aut \times \Gamma) \cup \set{\rej} \cup Q' \cup \bigcup\nolimits_{{\varphi_i \in \range(t)}}Q^i \enspace, \]
where the state $\rej$ is a rejecting sink.
The states from $Q^\aut \times \set{\ini, \fin}$ are used to simulate the original automaton at steps with stack height $0$~$(Q^\aut \times \set{\ini})$ and stack height at least $1$~$(Q^\aut \times \set{\fin})$, respectively.

For the sake of readability, we define the transition function for the different components of the automaton separately.
We also write $(\rightarrow, q)$ and $(\rightarrow_a, q)$ as shorthands for $(\rightarrow, q, \rej)$ and $(\rightarrow_a, \rej, q)$.
The easiest part of the transition function is that which controls the rejecting sink $\rej$, which is defined as $\delta_\mathit{sink}(\rej, a) \coloneq (\rightarrow, \rej)$ for all $a \in \Sigma$.

When encountering a final state of $\aut$, we need to be able to move to the successors of the initial states of $\aut'$ in order to model acceptance of $\aut$ on the finite prefix read so far.
To achieve a uniform presentation, we define the auxiliary formula
	$\chi^f(q, a) \coloneq \bigvee_{q_I' \in I'} \delta'(q_I', a)$ if $q \in F^\aut$ and $\chi^f(q, a) \coloneq (\rightarrow, \rej)$ otherwise.
	
Moreover, we need notation to denote transitions into the automata $\aut_i$ implementing the tests of $\aut$.
More precisely, since we only transition into these automata upon leaving the states labeled with the respective test, we need to transition into the successors of one of the initial states of the implementing automata.
To this end, we define the auxiliary formula $\theta_q^a \coloneq \bigvee_{q_I \in I^i} \delta^i(q_I, a)$, where $t(q) = \varphi_i$.

For local actions the main copy of the automaton can simply simulate the behavior of $\aut$ on the input word.
Hence we have
\begin{multline*}
 \delta_\mathit{main}((q, b), l) \coloneq \Big[ \chi^f(q, l) \lor \bigvee\nolimits_{{(q, l, q') \in \Delta}} (\rightarrow, (q', b)) \Big] \land \theta_q^l \\ \text{\quad  for } l \in \locals, b \in \set{\ini, \fin}
\end{multline*}

When reading a call, the automaton nondeterministically guesses whether it jumps to the matching return or whether it simulates the state transition while ignoring the effects on the stack.
In the former case, it guesses a transition $(q, c, q', A) \in \Delta$ and a state $q'' \in Q$, spawns a verification automaton verifying that it is possible to go from $q'$ to $q''$ by popping $A$ off the stack in the final transition, and continues at the matching return in state $q''$.
In the latter case it ignores the effects on the stack and denotes that it may not read any returns from this point onwards by setting the binary flag in its state to $\fin$.
\begin{multline*}
	\delta_\mathit{main}((q, b), c)\coloneq \Big[ \chi^f(q) \lor \\
	\bigvee\nolimits_{{(q, c, q', A) \in \Delta, q'' \in Q}} \Big[ (\rightarrow, (q', q'', A)) \land (\rightarrow_a, (q'', b)) \Big] \lor \\
	\bigvee\nolimits_{{(q, c, q', A) \in \Delta}} (\rightarrow, (q', \fin)) \Big] \land \theta_q^c \text{\quad  for } c \in \calls, b \in \set{\ini, \fin}
\end{multline*}

The main automaton may only handle returns as long as it has not skipped any calls.
If it encounters a return after having skipped a push action, it rejects the input word, since the return falsifies its earlier guess of an unmatched call.
\begin{align*}
	&\delta_\mathit{main}((q, \ini), r) \coloneq \Big[ \chi^f(q, r) \lor \bigvee\nolimits_{{(q, r, \bot, q') \in \Delta}} (\rightarrow, (q', \ini)) \Big] \land \theta_q^r \text{\quad  for } r \in \returns \\
	&\delta_\mathit{main}((q, \fin), r) \coloneq (\rightarrow, \rej) \text{\quad  for } r \in \returns
\end{align*}

The transition function $\delta_\mathit{main}$ determines the behavior of the main automaton.
It remains to define the behavior of the verifying automata.
These behave similarly to the main automaton on reading local actions and calls.
The main difference in handling calls is that these automata do not need to guess whether or not a call is matched:
Since they are only spawned on reading supposedly matched calls and accept upon reading the matching return, all calls they encounter must be matched as well.
Additionally, they never transition to the automaton $\aut'$, but merely to the automaton implementing the test of the current state upon having verified their guess.
\begin{align*}
	\delta_\mathit{ver}((q, q', A), l) \coloneq & \Big[ \bigvee\nolimits_{{(q, l, q'') \in \Delta}} (\rightarrow, (q'', q', A)) \Big] \land \theta_q^l \text{\quad  if } l \in \locals \\
	\delta_\mathit{ver}((q, q', A), c) \coloneq & \Big[ \bigvee\nolimits_{{(q, c, q'', A') \in \Delta, q''' \in Q}} (\rightarrow, (q'', q''', A')) \land \\
	&\enspace\enspace (\rightarrow_a, (q''', q', A)) \Big] \land \theta_q^c \text{\quad  if } c \in \calls \\
	\delta_\mathit{ver}((q, q', A), r) \coloneq & \theta_q^r \text{\quad  if } r \in \returns, (q, r, A, q') \in \Delta \\
	\delta_\mathit{ver}((q, q', A), r) \coloneq & (\rightarrow, \rej) \text{\quad  if } r \in \returns, (q, r, A, q') \not\in \Delta
\end{align*}
We then define the complete transition function $\delta$ of $\aut_\varphi$ as the union of the previously defined partial transition functions.
Since their domains are pairwise disjoint, this union is well-defined.
\[ \delta \coloneq \delta_\mathit{sink} \cup \delta' \cup \bigcup\nolimits_{{\varphi_i \in \range(t)}} \delta^i \cup \delta_\mathit{main} \cup \delta_\mathit{ver} \]

The coloring of $\aut_\varphi$ is obtained by copying the coloring of $\aut'$ and the $\aut^i$ and by coloring all states resulting from the translation of $\aut$ with $1$.
Thus, we force every path of the run of $\aut$ to eventually leave this part of the automaton, since this automaton only accepts a finite prefix of the input word.
The \oneaja 
\[ \aut_\varphi \coloneq (Q, \vpdalphabet, \delta, I^\aut \times \set{\ini}, \col \cup \col' \cup \bigcup\nolimits_{{\varphi_i \in \range(t^\aut)}}\col^i) \]
then recognizes the language of $\varphi = \ddiamond{\aut}\varphi'$, where $\col: q \mapsto 1$ for all $q \in (Q^\aut \times \{\ini, \fin\}) \cup (Q^\aut \times Q^\aut \times \Gamma) \cup \set{\rej}$.
\end{proof}

By combining Lemmas \ref{lem:aja-to-nvpa} and \ref{lem:vldl-to-aja} we see that \bnvpas are at least as expressive as \vldl formulas.
This proves the direction from logic to automata of Theorem~\ref{thm:vldl-eq-vpa}.
The construction via \oneajas yields automata of exponential size in the number of states.
This blowup is unavoidable, which we show by relying on the analogous lower bound for translating \ltl into Büchi automata, obtained by encoding an exponentially bounded counter in \ltl.

\begin{lemma}
\label{lem:conciseness}
	There exists a pushdown alphabet $\vpdalphabet$ such that for all $n \in \nats$ there exists a language $L_n$ that is defined by a \vldl formula over $\vpdalphabet$ of polynomial size in $n$, but every \bnvpa over $\vpdalphabet$ recognizing $L_n$ has at least exponentially many states in $n$.
\end{lemma}

\begin{proof}
	We use the pushdown alphabet $\vpdalphabet = (\calls, \returns, \locals) = (\emptyset, \emptyset, \{0, 1, \#\})$.
For any $n \in \nats$ and any $i \in [0;2^n-1]$ we write $\bin{i}_n$ to denote the binary encoding of $i$ using~$n$ bits.
Moreover, we define the language $L_n \coloneq \set{ \# \bin{0}_n \# \cdots \# \bin{2^n-1}_n\#^\omega}$, which only contains a single word encoding an~$n$-bit counter.
It is known that there exists an \ltl formula of polynomial length in $n$ that defines~$L_n$.
Thus, there also exists a \vldl formula of polynomial length defining this language due to Lemma~\ref{lem:ldl-to-vldl}.

Furthermore, since all symbols are local actions, any \bnvpa recognizing $L_n$ cannot use its stack and thus has to work like a traditional finite automaton with Büchi acceptance.
Again, it is known that all Büchi automata recognizing~$L_n$ have at least exponentially many states in~$n$.
Consequently, all \bnvpas recognizing $L_n$ have at least exponentially many states in $n$.
\end{proof}

After having shown that \vldl has the same expressiveness as \bnvpas, we now turn our attention to several decision problems for this logic.
Namely, we study the satisfiability and the validity problem, as well as the model checking problem.
Moreover, we consider the problem of solving visibly pushdown games with \vldl winning conditions.

\section{Satisfiability and Validity are \exptime-complete}
\label{sec:satisfiability-validity}

\newcommand{\conf}[0]{\mathit{Conf}}

We say that a \vldl formula $\varphi$ is satisfiable if it has a model.
Dually, we say that $\varphi$ is valid if all words are models of $\varphi$.
Instances of the satisfiability and validity problem consist of a \vldl formula $\varphi$ and ask whether $\varphi$ is satisfiable and valid, respectively.
Both problems are decidable in exponential time.
We also show both problems to be \exptime-hard.

\begin{theorem}
\label{thm:satisfiability-exptime-complete}
The satisfiability problem and the validity problem for \vldl are \exptime-complete.	
\end{theorem}

\begin{proof}
Due to duality, we only show \exptime-completeness of the satisfiability problem.
Membership follows from the \oneaja-emptiness-problem being in \exptime \cite{bozzelli07} and Lemma~\ref{lem:vldl-to-aja}.

It remains to show~$\exptime$-hardness, which we prove by prociding a reduction from the word problem for polynomially space-bounded alternating Turing machines.
This problem asks whether a given word is accepted by a given alternating Turing machine.
Since a run of an alternating Turing machine is a finite tree, it can be serialized as a word, where the subtrees are delimited by special symbols.
Such a word can then be checked for correctly encoding some tree using the stack of a VPA.
Adherence to the transition relation, as well as the property that the tree describes an accepting run of the Turing machine can be checked mostly locally without using the stack.
These constraints can be expressed in \vldl, such that their conjunction is satisfiable  if, and only if, there exists an accepting run of the Turing machine on the word, i.e., if the Turing machine accepts the word.

An alternating Turing machine (ATM) \cite{chandra76} $\tm = (Q_\exists, Q_\forall, \Gamma, q_I, \Delta, F)$ consists of
	two finite disjoint sets $Q_\exists$ and $Q_\forall$ of states, which are called existential and universal states, respectively, for which we write $Q \coloneq Q_\exists \cup Q_\forall$,
	a tape alphabet $\Gamma$ containing a blank symbol~$B$,
	an initial state $q_I \in Q \setminus F$,
	a transition relation~$\Delta\subseteq Q \times \Gamma\times Q \times \Gamma \times \set{L, R}$,
	and a set of final states $F \subseteq Q$.

Let $p(n)$ be some polynomial.
A configuration~$c$ of a $p(n)$-bounded ATM~$\tm$ on an input word~$w$ is a word of length $p(\card{w})+1$ over the alphabet $\Gamma \cup Q $ that contains exactly one symbol from $Q$. Let $\conf \coloneq \Gamma^* Q \Gamma^* \cap (Q \cup \Gamma)^{p(\card{w})+1}$ denote the set of such configurations. 
If $c \in \conf$ contains a symbol from $Q_\exists$ ($Q_\forall$), we call $c$ existential (universal). Furthermore, a transition~$(q, a, q', a', D) \in \Delta$ with $D \in \set{L,R}$ is existential (universal), if $q \in Q_\exists$ ($q \in Q_\forall$). 
We assume w.l.o.g.\ that every configuration has exactly two applicable transitions and that the initial state is not final.

A run of a $p(n)$-bounded ATM~$\tm$ on $w$ is a finite tree that is labeled with configurations of $\tm$ on $w$.
Each non-terminal vertex has either one or two successors, depending on whether it is labeled with an existential or a universal configuration. These successors have to be labeled by one or two successor configurations.
A run is accepting if all terminal vertices are labeled with final configurations.
An ATM $\tm$ accepts a word $w$ if there exists an accepting run of~$\tm$ on $w$.

An instance of the word problem consists of a $p(n)$-space-bounded ATM $\tm$ and a word $w$ and asks whether or not $\tm$ accepts $w$.
This problem is \exptime-hard~\cite{chandra76}.
	
\newcommand{\lefthash}{<}
\newcommand{\righthash}{>}
\newcommand{\pushconf}[1]{\mathit{push}(#1)}
\newcommand{\popconf}[1]{\mathit{pop}(#1)}
\newcommand{\enc}[1]{\mathit{enc}(#1)}
\newcommand{\hashes}[0]{\mathit{Tags}}

We encode runs of $\tm$ by linearizing them as words using tags of the form $\lefthash_\tau^i$ and $\righthash_\tau^i$ for $i \in \set{1,2}$ to delimit the encoding of the first and second subtree of a vertex (recall that we assume that every configuration has at most two successors).
Here, $\tau$ denotes the transition that is applied to obtain the configuration of the root of this subtree.
Moreover, we use the tags $\lefthash_\ell$ and $\righthash_\ell$ to denote leaves.

Formally, we define the pushdown alphabet~$\vpdalphabet = (\calls, \returns, \locals )$ with
\begin{itemize}
	\item $\calls = ((Q \cup \Gamma ) \times \set{\push\thinspace}) \cup \set{\lefthash_{\tau}^1 \mid \tau \in \Delta} \cup \set{\lefthash_\ell}$,

	\item $\returns = ((Q \cup \Gamma ) \times \set{\pop\thinspace}) \cup \set{\righthash_\tau^1 \mid \tau \text{ existential}} \cup \set{\righthash_\tau^2 \mid \tau \text{ universal}} \cup \set{\righthash_\ell}$, and
	
	\item $\locals = \set{\righthash_\tau^1, \lefthash_\tau^2 \mid \tau \text{ universal}} \cup \set{\#}$.
\end{itemize}
Let $\hashes = \set{\lefthash_\tau^1, \righthash_\tau^1 \mid \tau \text{ existential}} \cup \set{\lefthash_\tau^1, \righthash_\tau^1, \lefthash_\tau^2, \righthash_\tau^2 \mid \tau \text{ universal}} \cup \set{\lefthash_\ell, \righthash_\ell}$.

For $w = w_0\cdots w_n \in \conf^*$ and $d \in \set{\push\thinspace, \pop\thinspace}$, let $(w, d) \coloneq (w_0, d)\cdots(w_n, d)$, which we lift to languages in the obvious way.
Furthermore, let $w^r \coloneq w_n \cdots w_0$.
Let $c \in \conf$. We define $\pushconf{c} \coloneq (c, \push\thinspace)$ and $\popconf{c} \coloneq (c^r, \pop\thinspace)$

Using this, we encode a run of $\tm$ by recursively iterating over its vertices~$v$ as follows:
\begin{itemize}
	\item $\enc{v}\coloneq
			\lefthash_\ell \cdot \pushconf{c} \cdot 
			\righthash_\ell \cdot \popconf{c}$, if $v$ is a leaf labeled with the configuration~$c$.

	\item $\enc{v} \coloneq
		\lefthash_\tau^1 \cdot \pushconf{c} \cdot
		\enc{v_1} \cdot
		\righthash_\tau^1 \cdot \popconf{c}$,
			if $v$ has a single child~$v_1$, $v$ is labeled by the (existential) configuration~$c$, and $\tau$ is the transition that is applied to $c$ to obtain the label of $v_1$.  

	\item $\enc{v} \coloneq
		\lefthash_{\tau_1}^1 \cdot \pushconf{c} \cdot
		\enc{v_1} \cdot\
		\righthash_{\tau_1}^1 \cdot \popconf{c}
		\lefthash_{\tau_2}^2 \cdot \pushconf{c} \cdot
		\enc{v_2} \cdot
		\righthash_{\tau_2}^2 \cdot \popconf{c}
		$,
			if $v$ has two children~$v_1$ and $v_2$, $v$ is labeled by the (universal) configuration~$c$, and $\tau_i$, for $i \in \set{1,2}$, is the transition that is applied to $c$ to obtain the label of $v_i$.  
\end{itemize}
Thus, a complete run with root~$v$ is encoded by $\enc{v}\cdot \#^\omega$. Our goal is to construct a formula that is satisfied only by words that encode initial accepting runs of $\tm$ on $w$.
To this end, we need to formalize the following six conditions on an infinite word~$\alpha \in \Sigma^\omega$:

\begin{enumerate}
	\item $\alpha \in (\hashes \cdot \conf)^+ \cdot \#^\omega$ and begins with $\lefthash_\tau^1 \cdot\ (c_I, \push\thinspace)$, where $c_I$ is the initial configuration of $\tm$ on $w$ and where $\tau$ is a transition that is applicable to~$c_I$.
	
		\item Every $\lefthash_\tau^i$, $i \in \set{1,2}$, is directly followed by $(c, \push\thinspace)$ for some configuration~$c$ to which $\tau$ is applicable. Furthermore, say the stack height is $n$ after this infix. Then, we require that this stack height is reached again at a later position, and at the first such position, the infix~$\righthash_\tau^1 \cdot (c^r, \pop\thinspace)$ starts.

	\item Every $\righthash_\tau^1$ with universal~$\tau$, which is directly followed by  $(c^r, \pop\thinspace)$ for some configuration~$c$ (assuming the previous condition is satisfied), is directly followed by $(c^r, \pop\thinspace) \cdot \lefthash_{\tau'}^2 \cdot (c, \push\thinspace)$, where $\tau' \not= \tau$ is the unique other transition that is applicable to $c$.
	
	\item Every $\lefthash_\tau^i$, $i \in \set{1,2}$, is directly followed by $(c, \push\thinspace) < (c', \push\thinspace)$ for some $< \in \set{\lefthash_\tau^1 \mid \tau \in \Delta} \cup \set{\lefthash_\ell }$ such that $\tau$ is applicable to $c$ and $c'$ is the corresponding successor configuration.

	\item Every $\lefthash_\ell$ is directly followed by $(c, \push\thinspace) \righthash_\ell (c^r, \pop\thinspace)$ for some accepting configuration of $\tm$.

	\item Stack height zero has to be reached after a non-empty prefix, and from the first such position onwards, only $\#$ appears. 
\end{enumerate}

It is straightforward to come up with polynomially-sized \vldl formulas expressing these conditions (note that only the second and sixth condition require non-trivial usage of the stack). Furthermore, $\alpha$ satisfies the conjunction of these properties if, and only if, it encodes an accepting run of $\tm$ on $w$. Thus, as the word problem for polynomially space-bounded ATMs is \exptime-hard, the satisfiability problem for \vldl is \exptime-hard as well.
\end{proof}

\section{Model Checking is \exptime-complete}
\label{sec:model-checking}

We now consider the model checking problem for \vldl.
An instance of the model checking problem consists of a \vps $\vpsys$, an initial state $q_I$ of $\vpsys$, and a \vldl formula $\varphi$ and asks whether $\traces(\vpsys, q_I) \subseteq L(\varphi)$ holds true, where $\traces(\vpsys, q_I)$ denotes the set obtained by mapping each run of $\vpsys$ starting in~$q_I$ to the sequence of labels of the traversed edges.
This problem is decidable in exponential time due to Lemma~\ref{lem:vldl-to-aja} and an exponential-time model checking algorithm for \oneajas \cite{bozzelli07}.
Moreover, the problem is \exptime-hard, as it subsumes the validity problem.

\begin{theorem}
\label{thm:model-checking-hardness}
Model checking \vldl specifications against \vpss is \exptime-complete.
\end{theorem}

\begin{proof}
	Membership in \exptime follows from Lemma \ref{lem:vldl-to-aja} and the membership of the problem of checking visibly pushdown systems against \oneaja specifications in \exptime \cite{bozzelli07}.
	Moreover, since the validity problem for \vldl is \exptime-hard and since validity of $\varphi$ is equivalent to $\traces(\vpsys_\mathit{univ}) \subseteq \varphi$, where $\vpsys_\mathit{univ}$ with $\traces(\vpsys_\mathit{univ}) = \Sigma^\omega$ is effectively constructible in constant time, the model checking problem for \vldl is \exptime-hard as well.
\end{proof}

\section{Solving \vldl Games is \threeexp-complete}
\label{sec:realizability}

In this section we investigate visibly pushdown games with winning conditions given by \vldl formulas.
We consider games with two players, called Player~$0$ and Player~$1$, respectively.

A two-player game with \vldl winning condition $\game = (V_0, V_1, \Sigma, E, v_I, \ell, \varphi)$ consists of
two disjoint, at most countably infinite sets $V_0$ and $V_1$ of vertices, where we define $V \coloneq V_0 \cup V_1$,
a finite alphabet $\Sigma$,
a set of edges $E \subseteq V \times V$,
an initial state $v_I \in V$,
a labeling $\ell\colon V \rightarrow \Sigma$,
and a \vldl formula $\varphi$ over some partition of $\Sigma$,
called the winning condition.
	
A play $\pi = v_0v_1v_2\cdots$ of $\game$ is an infinite sequence of vertices of $\game$ with $(v_i, v_{i+1}) \in E$ for all $i \geq 0$.
The play $\pi$ is initial if $v_0 = v_I$.
It is winning for Player~$0$ if $\ell(v_1)\ell(v_2)\ell(v_3)\cdots$~\footnote{Note that the sequence of labels trace omits the label of the first vertex for technical reasons.} is a model of $\varphi$.
Otherwise $\pi$ is winning for Player~$1$.

A strategy for Player~$i$ is a function $\sigma\colon V^*V_i \rightarrow V$, such that $ (v, \sigma(w \cdot v)) \in E$ for all $v \in V_i$, $w \in V^*$.
We call a play $\pi = v_0v_1v_2\cdots$ consistent with $\sigma$ if $\sigma(\pi') = v_{n+1}$ for all finite prefixes $\pi' = v_0\cdots v_n$ of $\pi$ where $v_n \in V_i$.
A strategy~$\sigma$ is winning for Player~$i$ if all initial plays that are consistent with $\sigma$ are winning for that player.
We say that Player~$i$ wins $\game$ if she has a winning strategy.
If either player wins $\game$, we say that $\game$ is determined.

A visibly pushdown game (VPG) with a \vldl winning condition $\pdgame = (\vpsys, Q_0, Q_1, q_I, \varphi)$ consists of a VPS $\vpsys = (Q, \vpdalphabet, \Gamma, \Delta)$, a partition of $Q$ into~$Q_0$ and~$Q_1$, an initial state $q_I \in Q$, and a \vldl formula $\varphi$ over $\vpdalphabet$.
The VPG~$\pdgame$ then defines the two-player game
	$\game_\pdgame = (
		V_0,
		V_1,
		\Sigma,
		E,
		v_I,
		\ell,
		\varphi)$
	with $V_i \coloneq Q_i \times ((\Gamma \setminus \set{\bot})^*\cdot\bot) \times \Sigma$, $v_I = (q_I, \bot, a)$ for some $a \in \Sigma$ (recall that the trace disregards the label of the initial vertex), $((q, \gamma, a), (q', \gamma', a')) \in E$ if there is an $a'$-labeled edge from~$(q, \gamma)$ to $(q, \gamma')$ in the configuration graph $\graph_\vpsys$, and $\ell\colon (q, \gamma, a) \mapsto a$.
Solving a VPG $\pdgame$ means deciding whether Player~$0$ wins $\game_\pdgame$.

\begin{proposition}
\label{prop:vpg-determinacy}
VPGs with VLDL winning conditions are determined.
\end{proposition}

\begin{proof}
Since each \vldl formula defines a language in \omegavpl due to Theorem~\ref{thm:vldl-eq-vpa}, each VPG with \vldl winning condition is equivalent to a VPG with an \omegavpl winning condition. These are known to be determined \cite{loeding04}.	
\end{proof}

We show that solving VPGs with winning conditions specified in \vldl is harder than solving VPGs with winning conditions specified by \bnvpas.
Indeed, the problem is \threeexp-complete.
\begin{theorem}
\label{thm:vpg-completeness}
Solving VPGs with \vldl winning conditions is \threeexp-complete.	
\end{theorem}

\begin{proof}
We solve VPGs with \vldl winning conditions by first constructing a \bnvpa~$\aut_\varphi$ of exponential size from the winning condition $\varphi$ and by then solving the resulting visibly pushdown game with a \bnvpa winning condition~\cite{loeding04}.
As VPGs with \bnvpa winning conditions can be solved in doubly-exponential time, this approach takes triply-exponential time in~$\card{\varphi}$ and exponential time in~$\card{\vpsys}$.

We show \threeexp-hardness of the problem by a reduction from solving pushdown games with \ltl winning conditions, which is known to be \threeexp-complete~\cite{loeding04}.
A pushdown game with an \ltl winning condition $\pdgame = (\mathcal S, V_I, V_O, \psi)$ is defined similarly to a VPG, except for the relaxation that $\mathcal S$ may now be a traditional pushdown system instead of a visibly pushdown system.
Specifically, we have $\Delta \subseteq (Q \times \Gamma \times \Sigma \times Q \times \Gamma^{\leq 2})$, where $\Gamma^{\leq 2}$ denotes the set of all words over $\Gamma$ of at most two letters.
Stack symbols are popped off the stack using transitions of the form $(q, A, a, q', \epsilon)$, the top of the stack can be tested and changed with transitions of the form $(q, A, a, q', B)$, and pushes are realized with transitions of the form $(q, A, a, q', BC)$.
Additionally, the winning condition is given as an \ltl formula instead of a \vldl formula.
The two-player game $\game_\pdgame$ is defined analogously to the visibly pushdown game.

Since the pushdown game admits transitions such as $(q, A, a, q', BC)$, which pop $A$ off the stack and push $B$ and $C$ onto it, we need to split such transitions into several transitions in the visibly pushdown game.
We modify the original game such that every transition of the original game is modeled by three transitions in the visibly pushdown game, up to two of which may be dummy actions that do not change the stack.
As each transition may perform at most three operations on the stack, we can keep track of the list of changes still to be performed in the state space.
We perform these actions using dummy letters $c$ and~$l$, which we add to $\Sigma$ and read while performing the required actions on the stack.
We choose the vertices $V'_X = V_X \cup (V_X \times (\Gamma \cup \set{\#})^{\leq 2})$ and the alphabet $\vpdalphabet = (\set{c}, \Sigma, \set{l})$.

\begin{figure}[]
\begin{center}
\begin{tikzpicture}[xscale=2.2,yscale=2, thick]

\tikzstyle{p1}=[draw,circle,minimum size=1.3cm,scale=.8]
\tikzstyle{p2}=[draw,circle,minimum size=1.3cm,scale=.8]
\tikzstyle{transition}=[draw,-stealth]

	\node (pop-label) at (-.5,.2) {$(i)$};
	\node[p1] (pop-1) at (0,0) {$q$};
	\node[p1] (pop-2) at (1,0) {$q', \#\#$};
	\node[p1] (pop-3) at (2,0) {$q', \#$};
	\node[p2] (pop-4) at (3,0) {$q'$};
	
	\path[transition] (pop-1) edge node[anchor=south] {$a, \pop A$} (pop-2);
	\path[transition] (pop-2) edge node[anchor=south] {$l, \local$} (pop-3);
	\path[transition] (pop-3) edge node[anchor=south] {$l, \local$} (pop-4);
	
	\node (test-label) at (-.5,-.8) {$(\mathit{ii})$};
	\node[p1] (test-1) at (0,-1) {$q$};
	\node[p1] (test-2) at (1,-1) {$q', B\#$};
	\node[p1] (test-3) at (2,-1) {$q', \#$};
	\node[p2] (test-4) at (3,-1) {$q'$};
	
	\path[transition] (test-1) edge node[anchor=south] {$a, \pop A$} (test-2);
	\path[transition] (test-2) edge node[anchor=south] {$c, \push B$} (test-3);
	\path[transition] (test-3) edge node[anchor=south] {$l, \local$} (test-4);
	
	\node (pop-label) at (-.5,-1.8) {$(\mathit{iii})$};
	\node[p1] (push-1) at (0,-2) {$q$};
	\node[p1] (push-2) at (1,-2) {$q', BC$};
	\node[p1] (push-3) at (2,-2) {$q', B$};
	\node[p2] (push-4) at (3,-2) {$q'$};
	
	\path[transition] (push-1) edge node[anchor=south] {$a, \pop A$} (push-2);
	\path[transition] (push-2) edge node[anchor=south] {$c, \push C$} (push-3);
	\path[transition] (push-3) edge node[anchor=south] {$c, \push B$} (push-4);
	
\end{tikzpicture}	
\end{center}
	\caption{Construction of a VPG from a pushdown game for transitions of the forms $(i)$ $(q, a, A, q', \epsilon)$,
	$(\mathit{ii})$ $(q, a, A, q', B)$, and
	$(\mathit{iii})$ $(q, a, A, q', BC)$.}
	\label{fig:vpg-hardness}
\end{figure}
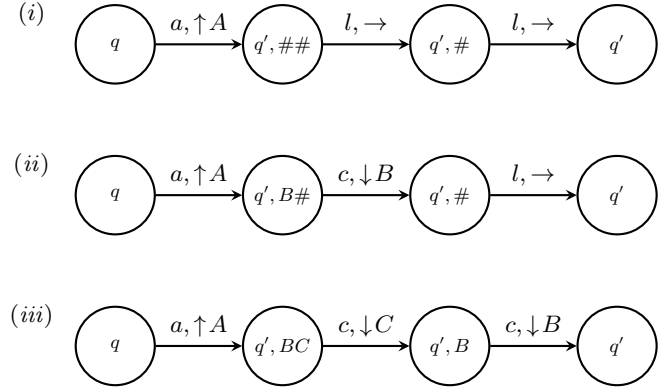

We transform $\pdgame$ as shown in Figure~\ref{fig:vpg-hardness} and obtain the VPG $\pdgame'$.
Moreover, we transform the winning condition~$\psi$ of $\pdgame$ into $\psi'$ by inductively replacing each occurrence of $\ltlnext\psi$ by $\ltlnext^3\psi'$ and each occurrence of $\psi_1\ltluntil\psi_2$ by $(\psi'_1 \lor c \lor l)\ltluntil (\psi'_2 \land \neg c \land \neg l)$.
We subsequently translate the resulting \ltl formula $\psi'$ into an equivalent \vldl formula $\varphi$ using Lemma~\ref{lem:ldl-to-vldl}.
The input player wins $\pdgame'$ with the winning condition $\varphi$ if and only he wins $\pdgame$ with the winning condition $\psi$.
Hence, solving VPGs with \vldl winning conditions is \threeexp-hard.	
\end{proof}

Moreover, Löding et al. have shown that in a visibly pushdown game with a winning condition given by a \bnvpa, Player~$0$, in general, requires infinite memory in order to win~\cite{loeding04}.
Thus, there is no winning strategy for her that is implemented by a finite automaton with output.
Such automata which are sufficient, e.g., for omega-regular games on finite graphs.
As we can translate \bnvpas into \vldl formulas, we obtain the same lower bound for VPGs with \vldl winning conditions.

Moreover, for each VPG~$\game$ with winning condition~$\varphi$, we can easily construct the game~$\game'$ by exchanging the states of Player~$0$ and Player~$1$ and obtain that Player~$0$ wins~$\game$ with winning condition~$\varphi$ if, and only if, she loses~$\game'$ with winning condition~$\neg\varphi$.
Hence, Player~$1$ requires, in general, infinite memory as well in order to win a VPG with \vldl winning condition.

\begin{corollary}
	There exists a VPG~$\game$ with \vldl winning condition such that Player~$0$ wins~$\game$, but requires infinite memory to do so.
	Similarly, there exists a VPG~$\game'$ with \vldl winning condition such that Player~$1$ win~$\game'$, but requires infinite memory to do so.
\end{corollary}

\section{Pushdown Linear Dynamic Logic}
\label{sec:pushdown-ldl}

In this work we extended \ldl by replacing the regular languages used as guards for the temporal operators by visibly pushdown languages.
We obtain an even stronger logic by using more expressive languages as guards, e.g., deterministic pushdown languages, which have deterministic pushdown automata (DPDA) as their canonical acceptors.
However, all relevant decision problems for the resulting logic called Deterministic Pushdown Linear Dynamic Logic (\dpldl) are undecidable, most importantly the satisfiability problem.

\begin{theorem}
\label{thm:dpldl-sat-undecidability}
	The satisfiability problem for \dpldl is undecidable.
\end{theorem}

\begin{proof}
	We reduce the problem of deciding nonemptiness of the intersection of two DPDA, which is known to be undecidable \cite{hopcroft01}, to the satisfiability problem for \dpldl.
	Let $\aut_1$ and $\aut_2$ be two DPDA over a shared alphabet $\Sigma$, pick $\# \notin \Sigma$ and consider $\varphi \coloneq \ddiamond{\aut_1}\# \land \ddiamond{\aut_2}\#$.
	Then $\varphi$ is satisfiable  if, and only if, $L(\aut_1) \cap L(\aut_2) \neq \emptyset$.
	Hence satisfiability of \dpldl is undecidable.	
\end{proof}

As the satisfiability problem reduces to model checking and to solving pushdown games with \dpldl winning conditions, both problems are also undecidable.

\begin{corollary}
	The validity problem and the problem of checking \dpldl specifications against \vpss as well as the problem of solving pushdown games against \dpldl winning conditions are undecidable.
\end{corollary}

Since every DPDA is also a PDA, the extension of \dpldl by nondeterministic pushdown automata inherits these undecidability results from \dpldl.
Thus, \vldl is, to the best of our knowledge, the most expressive logic that combines the temporal modalities of \ldl with guards specified by languages over finite words and still has decidable decision problems.

\section{Conclusion}
\label{sec:conclusion}

We have introduced Visibly Linear Dynamic Logic (\vldl) which strengthens Linear Dynamic Logic (\ldl) by replacing the regular languages used as guards in the latter logic with visibly pushdown languages.
\vldl characterizes the class of $\omega$-visibly pushdown languages.
We have provided effective translations from \vldl to \bnvpa and vice versa with an exponential blowup in size in both directions.
From automata to logic, this blowup cannot be avoided while it remains open whether or not it can be avoided in the other direction.

Figure~\ref{fig:constructions} gives an overview over the known formalisms that capture \omegavpl and the translations between them.
Our constructions are marked by solid lines, all others by dotted lines.
All constructions are annotated with the blowup they incur.

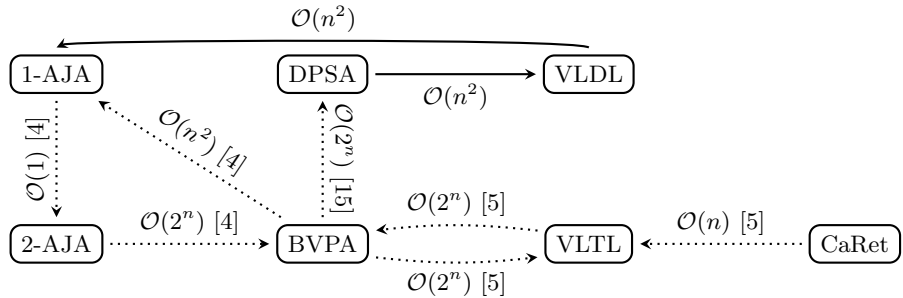
\begin{figure}[h]
\centering
\begin{tikzpicture}[xscale=3.5,yscale=2.25, thick]

\tikzstyle{construction}=[draw,-stealth]
\tikzstyle{self}=[solid]
\tikzstyle{others}=[dotted]

	\node[rectangle,rounded corners,draw,inner sep=4pt,outer sep=3pt] (caret) at (3,0) {\caret};
	\node[rectangle,rounded corners,draw,inner sep=4pt,outer sep=3pt] (vltl) at (2,0) {\vltl};
	\node[rectangle,rounded corners,draw,inner sep=4pt,outer sep=3pt] (bnvpa) at (1,0) {\bnvpa};
	\node[rectangle,rounded corners,draw,inner sep=4pt,outer sep=3pt] (dpsa) at (1,1) {DPSA};
	\node[rectangle,rounded corners,draw,inner sep=4pt,outer sep=3pt] (vldl) at (2,1) {\vldl};
	\node[rectangle,rounded corners,draw,inner sep=4pt,outer sep=3pt] (1-aja) at (0,1) {\oneaja};
	\node[rectangle,rounded corners,draw,inner sep=4pt,outer sep=3pt] (2-aja) at (0,0) {2-AJA};
	
	\path[construction,others] (caret) edge [transform canvas={xshift=0cm}] node [anchor=south,transform canvas={xshift=0cm}] {$\mathcal O(n)$ \cite{bozzelli14b}} (vltl);
	\path[construction,others] (vltl) edge [transform canvas={yshift=.1cm},bend right=10] node [anchor=south] {$\mathcal O(2^n)$ \cite{bozzelli14b}} (bnvpa);
	\path[construction,others] (bnvpa) edge [transform canvas={yshift=-.1cm},bend right=10] node [anchor=north] {$\mathcal O(2^n)$ \cite{bozzelli14b}} (vltl);
	\path[construction,others] (bnvpa) edge [] node [rotate=-90,transform canvas={xshift=.25cm}] {$\mathcal O(2^n)$ \cite{loeding04}} (dpsa);
	\path[construction,self] (dpsa) edge [transform canvas={xshift=0cm}] node [anchor=north] {$\mathcal O(n^2)$} (vldl);
	\path[construction,self] (vldl.north) to [out=135,in=35,looseness=0.1] node [anchor=south] {$\mathcal O(n^2)$} (1-aja.north);
	\path[construction,others] (1-aja) edge [transform canvas={xshift=0cm}] node [rotate=90,transform canvas={xshift=-.25cm}] {$\mathcal O(1)$ \cite{bozzelli07}} (2-aja);
	\path[construction,others] (2-aja) edge [transform canvas={xshift=0cm}] node [anchor=south] {$\mathcal O(2^n)$ \cite{bozzelli07}} (bnvpa);
	\path[construction,others] (bnvpa) edge [transform canvas={xshift=0cm}] node [anchor=south,rotate=-33.69] {$\mathcal O(n^2)$ \cite{bozzelli07}} (1-aja);
	
\end{tikzpicture}

	\caption{Formalisms capturing (subsets of) \omegavpl and translations between them.}
	\label{fig:constructions}
\end{figure}

In particular, there exist translations between \vltl and \vldl via \bnvpas that incur a doubly-exponential blowup in both directions, as shown in Figure~\ref{fig:constructions}.
In spite of this blowup the satisfiability problem and the model checking problem for both logics are \exptime-complete.
It remains open whether there exist efficient translations between the two logics.

We showed the satisfiability and the emptiness problem for \vldl, as well as model checking visibly pushdown systems against \vldl specifications, to be \exptime-complete.
Also, we proved that solving visibly pushdown games with \vldl winning conditions is \threeexp-complete.

Extending \vldl by replacing the guards with a more expressive family of languages quickly yields undecidable decision problems.
In fact, using deterministic pushdown languages as guards already renders all decision problems discussed in this work undecidable.

In contrast to \ldl \cite{vardi11} and \vltl \cite{bozzelli14b}, \vldl uses automata to define guards instead of regular or visibly rational expressions.
We are currently investigating a variant of \vldl where the \vpas guarding the temporal operators are replaced by visibly rational expressions (with tests), which is closer in spirit to \ldl. 

Moreover, the algorithm for checking satisfiability of \vldl formulas presented in this work relies on a translation of these formulas into \oneaja.
While the emptiness problem for \oneaja is known to be solvable in exponential time, no implementation of an algorithm doing so exists.
Thus, we are working on an alternative translation of \oneaja into alternating tree automata that preserves emptiness.
The resulting tree automaton is of exponential size and its emptiness problem is easily reducible to a parity game~\cite{FijalkowPinchinatSerre13}.
Since parity games serve as a backend for solving a multitude of verification problems, they have received considerable attention and mature tool support~\cite{FriedmannLange09,Keiren09}.
Such a construction could then be used for an implementation of a satisfiability checker for \vldl.
Moreover, it would make the problem amenable to the methods of bounded synthesis~\cite{FinkbeinerSchewe13}, thus providing a feasible approach to solving the problem in spite of its prohibitive worst-case complexity.


\paragraph{Acknowledgements}
The authors would like to thank Laura Bozzelli for providing the full version of \cite{bozzelli07} and Christof Löding for pointing out the \threeexp-hardness of solving infinite games for visibly pushdown games against \ltl specifications.

\bibliographystyle{splncs03}
\bibliography{main}

\begin{thebibliography}{10}
\providecommand{\url}[1]{\texttt{#1}}
\providecommand{\urlprefix}{URL }

\bibitem{alur04b}
Alur, R., Etessami, K., Madhusudan, P.: A temporal logic of nested calls and
  returns. In: {TACAS 2004}. LNCS, vol. 2988, pp. 467--481. Springer (2004)

\bibitem{alur04}
Alur, R., Madhusudan, P.: Visibly pushdown languages. In: {STOC 2004}. pp.
  202--211. ACM (2004)

\bibitem{bouajjani97}
Bouajjani, A., Esparza, J., Maler, O.: Reachability analysis of pushdown
  automata: Application to model-checking. In: Mazurkiewicz, A., Winkowski, J.
  (eds.) CONCUR 1997. LNCS, vol. 1243, pp. 135--150. Springer (1997), {Full
  version available at
  \url{http://www.liafa.univ-paris-diderot.fr/~abou/BEM97.pdf}}

\bibitem{bozzelli07}
Bozzelli, L.: Alternating automata and a temporal fixpoint calculus for visibly
  pushdown languages. In: Caires, L., Vasconcelos, V.T. (eds.) CONCUR 2007.
  LNCS, vol. 4703, pp. 476--491. Springer (2007)

\bibitem{bozzelli14b}
Bozzelli, L., S{\'a}nchez, C.: Visibly linear temporal logic. In: IJCAR 2014.
  LNCS, vol. 8562, pp. 418--483 (2014)

\bibitem{chandra76}
Chandra, A., Stockmeyer, L.: Alternation. In: FOCS 1976. pp. 98--108. IEEE
  (1976)

\bibitem{chen02}
Chen, H., Wagner, D.: {MOPS}: an infrastructure for examining security
  properties of software. In: Atluri, V. (ed.) CCS 2002. pp. 235--244. ACM
  (2002)

\bibitem{FaymonvilleZimmermann17}
Faymonville, P., Zimmermann, M.: Parametric linear dynamic logic. Inf. Comput.
  253,  237--256 (2017)

\bibitem{FijalkowPinchinatSerre13}
Fijalkow, N., Pinchinat, S., Serre, O.: Emptiness of alternating tree automata
  using games with imperfect information. In: Seth, A., Vishnoi, N.K. (eds.)
  {FSTTCS}. LIPIcs, vol.~24, pp. 299--311. Schloss Dagstuhl - LZI (2013)

\bibitem{FinkbeinerSchewe13}
Finkbeiner, B., Schewe, S.: Bounded synthesis. {STTT}  15(5-6),  519--539
  (2013)

\bibitem{FriedmannLange09}
Friedmann, O., Lange, M.: The {PGSolver} collection of parity game solvers.
  Tech. rep., University of Munich (2009), available at
  \url{github.com/tcsprojects/pgsolver/blob/master/doc/pgsolver.pdf}.

\bibitem{hopcroft01}
Hopcroft, J., Motwani, R., Ullman, J.: {Introduction to Automata Theory,
  Languages, and Computation}. Addison-Wesley (2001)

\bibitem{Keiren09}
Keiren, J.: An experimental study of algorithms and optimisations for parity
  games, with an application to Boolean Equation Systems. Master's thesis,
  Eindhoven University of Technology (2009)

\bibitem{leucker07}
Leucker, M., Sanch{\'e}z, C.: Regular linear temporal logic. In: Jones, C.B.,
  Liu, Z., Woodcock, J. (eds.) ICTAC 2007. pp. 291 -- 305. No. 4711 in LNCS
  (2007)

\bibitem{loeding04}
L{\"o}ding, C., Madhusudan, P., Serre, O.: Visibly pushdown games. In: Lodaya,
  L., Mahajan, M. (eds.) FSTTCS 2004. LNCS, vol. 3328, pp. 408--420. Springer
  (2005)

\bibitem{pnueli77}
Pnueli, A.: The temporal logic of programs. In: FOCS 1977. pp. 46--57. IEEE
  (1977)

\bibitem{vardi11}
Vardi, M.: The rise and fall of {LTL}. In: D'Agostino, G., Torre, S.L. (eds.)
  EPTCS 54 (2011)

\bibitem{vardi94}
Vardi, M., Wolper, P.: Reasoning about infinite computations. Inf. and Comp.
  115,  1--37 (1994)

\bibitem{wolper83}
Wolper, P.: Temporal logic can be more expressive. Inf. and Cont.  56,  72--99
  (1983)

\end{thebibliography}

\end{document}